\documentclass[%
 reprint,
superscriptaddress,
 amsmath,amssymb,
 aps,
]{revtex4-2}

\usepackage{graphicx}% Include figure files
\usepackage{dcolumn}% Align table columns on decimal point
\usepackage{bm}% bold math
\usepackage{hyperref}% add hypertext capabilities
\usepackage{amsmath,mathtools,amsthm}
\usepackage{cleveref}
\usepackage{xcolor}
\usepackage{braket}
\usepackage{needspace}
\newtheorem{theorem}{Theorem}
\newtheorem{proposition}{Proposition}
\newtheorem{corollary}{Corollary}[theorem]

\theoremstyle{definition}

\newtheorem{observation}{Observation}
\newtheorem{definition}{Definition}

\begin{document}

\preprint{APS/123-QED}

\title{Principles of Quantum Optimization for Constrained Problems}% Force line breaks with \\

\author{Einar Gabbassov}
\email{egabbass@uwaterloo.ca}
\affiliation{Department of Applied Mathematics, University of Waterloo, Waterloo, Ontario, N2L 3G1, Canada}
\affiliation{Institute for Quantum Computing, University of Waterloo, Waterloo, Ontario, N2L 3G1, Canada}
\affiliation{Perimeter Institute for Theoretical Physics, Waterloo, Ontario, N2L 2Y5, Canada}

\author{Gurpahul Singh}
\affiliation{Perimeter Institute for Theoretical Physics, Waterloo, Ontario, N2L 2Y5, Canada}
\affiliation{Department of Physics and Astronomy, University of Waterloo, Waterloo, Ontario N2L 3G1, Canada}

\author{Achim Kempf}
\affiliation{Department of Applied Mathematics, University of Waterloo, Waterloo, Ontario, N2L 3G1, Canada}
\affiliation{Institute for Quantum Computing, University of Waterloo, Waterloo, Ontario, N2L 3G1, Canada}
\affiliation{Perimeter Institute for Theoretical Physics, Waterloo, Ontario, N2L 2Y5, Canada}
\affiliation{Department of Physics and Astronomy, University of Waterloo, Waterloo, Ontario N2L 3G1, Canada}

\begin{abstract}
Constrained combinatorial optimization underlies many industrial and technological decision problems. We develop a spectral theory that unifies many quantum optimization algorithms. We show that computational slowdown is driven by entanglement restructuring: the creation, redistribution, and destruction of entanglement during system evolution. The severity of the slowdown depends on how much entanglement must be changed. We show that algebraic properties of constraints induce such restructuring, and that constraint-aware dynamics reduce the associated slowdown by avoiding unnecessary restructuring. This framework explains why constraint-aware quantum methods can outperform generic penalty-based approaches. The theory connects constrained optimization, computational complexity, entanglement dynamics, and Hamiltonian spectral structure across continuous-time and circuit-based quantum optimization paradigms.
\end{abstract}

\maketitle

%\tableofcontents

\section{\label{sec:intro}Introduction}
Linear programming (LP) and mixed integer programming (MIP) are foundational pillars of modern mathematical optimization \cite{bixby2012brief,strayer2012linear,Vanderbei2020LP,conforti2014integer}. They provide the standard mathematical language for constrained decision-making and underpin large parts of modern socio-technical infrastructure, including supply chains, transportation and logistics, energy systems, communication networks, and large-scale resource allocation. LP establishes a mathematically rigorous theory of constrained optimization: it studies optimality, the convex geometry induced by constraints, primal-dual formulations, optimality conditions, exact efficient algorithms, and computational complexity. MIP extends this framework to discrete decision variables, in which integrality constraints impose combinatorial structure and yield nonconvex solution sets.

By contrast, the quantum optimization literature has predominantly focused on a narrow subclass of MIP: unconstrained binary problems, most notably quadratic unconstrained binary optimization (QUBO) and equivalent Ising models. Standard benchmarks, such as MaxCut, naturally align with this problem class. This focus is largely driven by representational convenience: unconstrained binary problems map directly to the Hamiltonians of interacting spins \cite{lucas2014ising}. As a result, constrained problems are often first reformulated as QUBO problems, in which constraints are no longer treated as structural features of the problem but are instead enforced via energetic penalty terms added to the objective function \cite{gabbassov2025lagrangian,lucas2014ising,glover2018tutorial,de2024optimized}.

Quantum optimization research has only recently begun to develop algorithms that incorporate constraints directly, thereby capitalizing on their algebraic structure for more efficient search.  Such algorithms can typically distinguish candidate solutions both by objective value and by constraint satisfaction, thereby targeting problem classes closer to practical LP and MIP formulations.

Notable directions for incorporating constraints into quantum algorithms include feasibility-preserving mixer Hamiltonians, constraint-aware state preparation, coherent constraint checking using oracle subroutines, and repeated-measurement schemes that enable penalty-free exploration of feasible subspaces. For example, \cite{hadfield2019quantum,wang2020xy,cook2020quantum} introduced Hamiltonians, also called mixer operators, that move amplitude only within the feasible subspace for certain combinatorial optimization problems. \cite{fuchs2022constraint} generalizes the construction of feasibility-preserving mixers and discusses efficient implementations with minimal two-qubit gate counts, while \cite{hao2026constraint} introduces adaptive quantum routines that dynamically select feasibility-preserving mixers. In \cite{bartschi2020grover}, an equal superposition over all feasible solutions is prepared and then combined with a Grover-like mixer for amplitude amplification. In \cite{bucher2025penalty, bucher2025efficient}, the authors combine quantum phase estimation (QPE) with the quantum approximate optimization algorithm (QAOA): QPE is used to coherently detect constraint violations and, conditioned on the result, apply a corresponding QAOA cost layer. In the measurement-based approaches, \cite{herman2023constrained} repeatedly projects the evolving state onto the feasible subspace, thereby realizing a feasibility-preserving mixer. Each constraint is coherently evaluated into an auxiliary qubit, whose measurement implements the feasible-subspace projector; repeating this procedure suppresses transitions out of the feasible subspace through quantum Zeno dynamics. On the other hand, \cite{pawlak2023quantum,fuchs2024lx} uses stabilizer measurements and feedback to keep the dynamics within a feasible subspace. In \cite{pawlak2023quantum}, constraint-satisfying states are stabilized by syndrome measurements; detected violations are then corrected by feedback operations that recover the state in the feasible subspace.

As quantum optimization shifts toward such constrained problems, developing a general and physically principled understanding of how constraints affect quantum dynamics is essential. From a quantum perspective, constraints are not just modelling details: they can change the problem’s complexity class, reshape the geometry of the feasible Hilbert space, and modify the Hamiltonian spectral structure. Most importantly, as we will see, constraints determine how entanglement must be created, redistributed, or destroyed during the computation. Each of these effects can directly influence the quantum computational complexity and resource requirements.

In this work, we develop an implementation-independent theory of constrained quantum optimization. Specifically, our theory encompasses the aforementioned algorithms and more generally quantum optimization algorithms that are inspired by adiabatic \cite{farhiquantum,farhi2001quantum} or deliberately diabatic evolution \cite{feinstein2025robustness,muthukrishnan2016tunneling,fry2021locally}. This includes continuous-time adiabatic protocols, their discretized/digital gate-based counterparts, variational and projective-measurement approaches with parameterized quantum circuits, and quantum annealing heuristics.

Central to this theory is our rigorous demonstration that entanglement restructuring (the creation, redistribution, and destruction of entanglement among subsystems) during system evolution is the principal driver of computational slowdown. We show that the algebraic properties of problem constraints induce this process and that even a modest amount of entanglement can impose a severe computational cost when its structure must change rapidly.

This perspective changes how one should think about quantum optimization for constrained problems. We will show that the algebraic structure of constraints induces special sequences of narrow or closed spectral gaps. These gaps, in turn, allow constraint-aware dynamics to avoid unnecessary entanglement restructuring. Avoiding such restructuring reduces computational bottlenecks and enables more efficient algorithms.

Building on this insight, our theory provides design principles to reduce runtime bottlenecks by controlling entanglement dynamics arising from problem constraints.

The remainder of this work is organized as follows:
\begin{itemize}
    \item \Cref{sec:fund_classical_opt} introduces the principles of linear and integer programming. It explains how classical methods exploit the geometry and combinatorial structure induced by constraints.

    \item \Cref{sec:princip_quant_opt} introduces a common Hamiltonian framework for constrained quantum optimization. It distinguishes penalty-based methods, penalty-free methods, and relaxed mixers according to how they enforce feasibility and explore the search space.

    \item \Cref{sec:entanglement_restructuring_and_slowdown} relates computational slowdown to entanglement restructuring. It explains why the amount of entanglement or the spectral gap alone does not identify the computational difficulty.

    \item \Cref{sec:case_study} illustrates how entanglement restructuring informs quantum algorithm design. It shows that penalty terms can induce severe slowdown, whereas an overly restrictive feasibility-preserving mixer can make the optimal solution dynamically inaccessible. A richer mixer restores accessibility while keeping restructuring minimal, thereby enabling faster computation.

    \item \Cref{sec:entanglement_dynamics} develops the general theory of entanglement dynamics. It establishes conditions under which approaching energy levels exchange their eigenvectors and shows how successive eigenvector swaps produce entanglement restructuring.

    \item \Cref{sec:spectral_gaps} applies this theory to constrained quantum optimization. It shows that constraints induce sequences of exact or narrowly avoided level crossings, proves that fast evolution can jump the narrow crossings while preserving the required eigenvector, and explains why an exact global gap closing need not be a computational bottleneck.

    \item \Cref{sec:spectral_duality} establishes a spectral duality between penalty-free and penalty-based methods. In penalty-free methods, the relevant energy level ascends toward an excited eigenstate of the problem Hamiltonian; in penalty-based methods, it descends toward the ground energy level.

    \item \Cref{sec:state_transitions} derives the effective state transitions induced by the two approaches. Penalty-based transitions are weighted by constraint-violation penalties, whereas penalty-free transitions are weighted by the objective values of infeasible intermediate states.

    \item Finally, \Cref{sec:discussion} summarizes the resulting design principles and outlines directions for developing constraint-aware quantum algorithms.
\end{itemize}

\section{\label{sec:fund_classical_opt}Principles of Classical Optimization as Constraint-Aware Search}
Classical optimization owes much of its theoretical and practical success to treating constraints as central mathematical objects. The theories of linear and integer programming exploit the algebraic, geometric, and combinatorial structure induced by constraints to characterize the feasible search space, provide optimality guarantees, derive useful bounds, and design powerful algorithms. This progress made important problem classes once regarded as computationally difficult, such as linear programs, solvable in polynomial time. Digital computers enabled the resulting algorithms to be applied at scale to problems involving millions of variables and constraints \cite{bixby2012brief}. This section develops this classical viewpoint through a concise introduction to standard methods from operations research and mathematical optimization. The discussion highlights important ideas that merit deeper study and suggests adapting their underlying principles to the development of quantum algorithms for constrained problems.

For readers interested in the remarkable history of the development of mathematical optimization and its role in world history, we provide a brief overview in \Cref{ap:history}.

\subsection{Linear Programs}
Linear programming provides the simplest setting in which constraints, feasible regions, objective functions, optimality, and algorithmic search can be understood geometrically and mathematically. For this reason, linear programs form a natural starting point for understanding constrained optimization before moving to its quantum analog.

We begin by defining the objective function to be minimized. In LP, the objective is a multivariate linear function
\begin{equation}
    f(x) = c^{\mathsf T} x,
\end{equation}
where $c \in \mathbb{R}^n$ is a vector of coefficients. The variables $x \in \mathbb{R}^n_+$ must satisfy $m$ linear inequality or equality constraints
\begin{align*}
    a_{11} \, x_1 + \dots + a_{1n} \, x_n &\geq b_1 \\
    &\vdots \\
    a_{m1} \, x_1 + \dots + a_{mn} \, x_n &\geq b_m.
\end{align*}
The system of linear inequalities above is usually written in a compact form as
\begin{equation}
    A x \geq b,
\end{equation}
where the matrix $A \in \mathbb{R}^{m\times n}$ and the vector $b \in \mathbb{R}^m$ encode $m$ linear inequality constraints.

Then the canonical form of an LP problem can be written as
\begin{equation}\label{eq:canonical_lp}
    \begin{aligned}
        \min \quad & f(x)\\
        \text{subject to}\quad & Ax \ge b \\
        & \ \ x \ge 0
    \end{aligned}
\end{equation}
The feasible region associated with \cref{eq:canonical_lp} is
\begin{equation}\label{eq:feasible_set}
\mathcal{D} = \{\, x \in \mathbb{R}^n : Ax \ge b,\ x \ge 0 \,\}.
\end{equation}
It is worth pointing out two important properties of the set $\mathcal{D}$:
\begin{enumerate}
    \item \textbf{Polyhedral convexity.} $\mathcal{D}$ is a polyhedral convex set \cite[Ch.~3]{strayer2012linear}, meaning that it is enclosed by flat facets lying in hyperplanes where constraints are active, such as $(Ax)_i=b_i$ (and, in canonical form, $x_j=0$). Geometrically, $\mathcal{D}$ is a polyhedron; see \Cref{fig:polyhedron_and_graph} (a). This geometry enables polynomial-time algorithms for LP (e.g., interior-point methods \cite{potra2000interior}), placing LP in the class $\mathrm{P}$. In some cases, the polyhedron may be unbounded.
    \item \textbf{Extreme-point optimality.} If optimal solutions exist, then at least one extreme point of $\mathcal{D}$ is the optimal solution. An extreme point of a convex set is a point that is not ``in between" any two other distinct points of the set. In a polyhedron, extreme points are exactly the ``vertices" (corners). E.g., see \Cref{fig:polyhedron_and_graph} (b).
\end{enumerate}
Although $\mathcal{D}$ contains infinitely many feasible points, the number of extreme points is finite. A crude upper bound \cite[Ch.~5]{strayer2012linear} on the number of extreme points is
\begin{equation}\label{eq:vertex_bound}
\left | \{ \text{extreme points of $\mathcal{D}$}\} \right | \leq \binom{m+n}{n}.
\end{equation}
This bound suggests that we might want to consider only finitely many candidates at the vertices of the polyhedron.

\begin{figure}[t]
    \centering
    \includegraphics[scale=0.066]{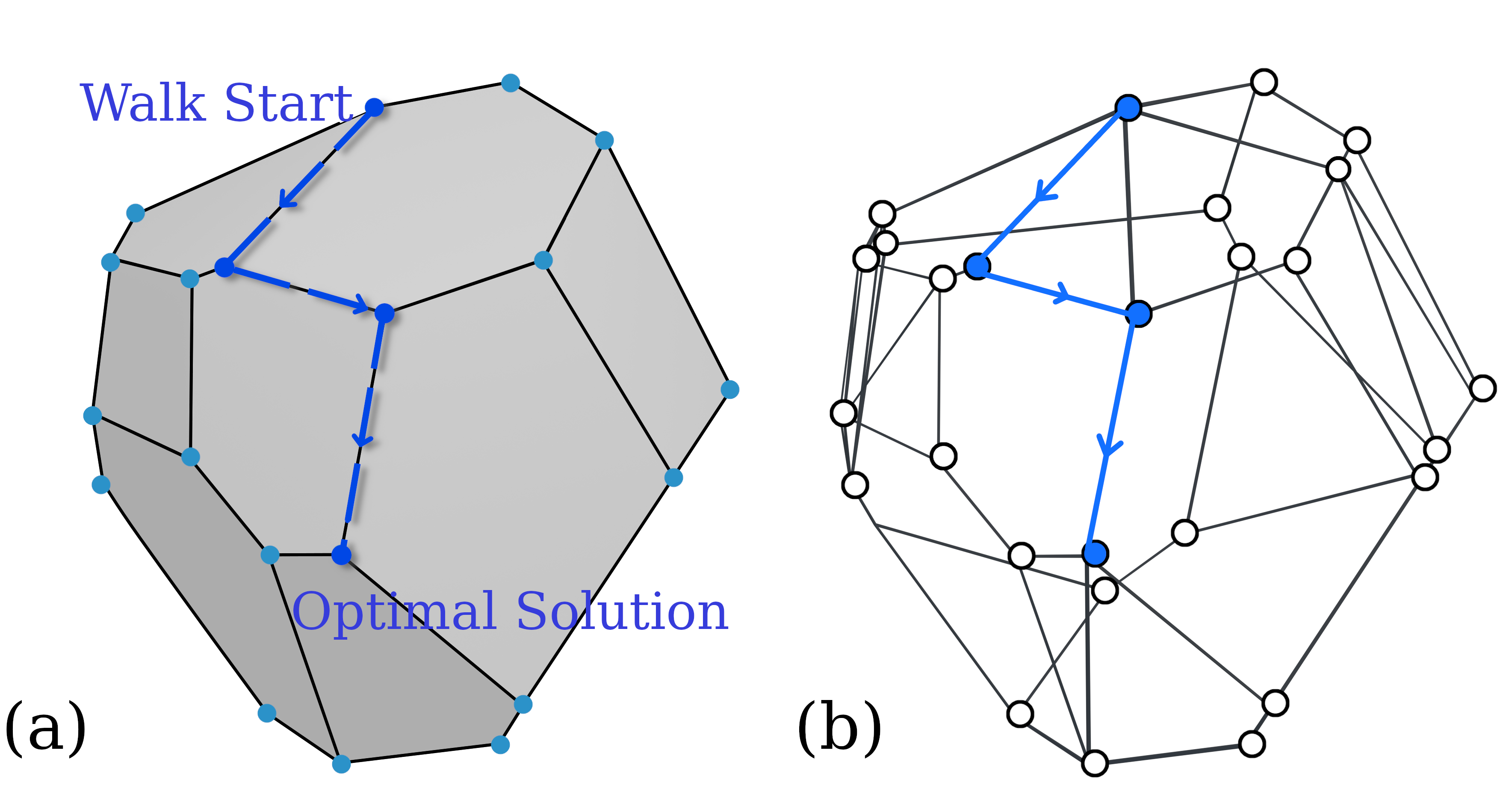}
    \caption{(a) A three-dimensional bounded polyhedron $\mathcal{D}$ defined by linear constraints. The interior and the boundaries of a polyhedron contain all feasible solutions. The optimal solution is situated on one of the vertices. The blue path is a walk on the vertices along the edges performed by the simplex algorithm. In the nondegenerate case, each vertex hop along the path decreases the objective function value. (b) The induced transition graph of the feasible polyhedron: vertices are feasible solutions, and edges represent allowable transitions between adjacent vertices.}
    \label{fig:polyhedron_and_graph}
\end{figure}

\subsection{Solving Linear Programs}\label{sec:solving_lp}
In this section, we give an intuitive geometric picture of how linear programs can be solved. If the LP \cref{eq:canonical_lp} is feasible and bounded, then there exists an optimal solution at an extreme point (vertex) of the polyhedral $\mathcal{D}$. This observation motivates a family of \emph{vertex-walk} methods: starting from a feasible vertex, one repeatedly moves along an edge of $\mathcal{D}$ to an adjacent vertex that improves the objective value, e.g., see \Cref{fig:polyhedron_and_graph}. This is the central idea behind the simplex algorithm. It proceeds by moving from vertex to adjacent vertex along the edges of the feasible polyhedron, with each step monotonically improving the objective value. However, in some cases, there may also be degenerate steps that cycle at the same vertex. Standard tie-breaking rules prevent cycling and ensure that the search continues \cite{strayer2012linear}.

While hopping from one vertex to the adjacent vertex yields an intuitive picture of the simplex's mechanics, it is important to appreciate that at any point of the computation, there does not exist an explicit description of the polyhedron geometry on which the walk is performed; i.e., no explicit knowledge of its vertices, faces and edges. The inner workings of the algorithm are both sophisticated and elegant, and mostly involve ingenious methods for working with systems of linear equations. See \cite{Vanderbei2020LP} for details. The upper bound in \cref{eq:vertex_bound} suggests that the walk can be fairly long. Indeed, the worst-case performance of the simplex algorithm is exponentially many vertex hops. Despite this worst-case behaviour, the simplex algorithm remains a standard method for solving LPs.

\subsection{0--1 Integer Programs}\label{sec:ip_problems}
In previous sections, we introduced LPs and a standard method of solving them. We now consider a more challenging class of problems obtained by imposing integrality constraints. In particular, we focus on binary ($0$--$1$) decision variables. A canonical form is the 0--1 integer program (IP):
\begin{equation}\label{eq:canonical_ip}
    \begin{aligned}
        \text{min } \quad & f(x)\\
        \text{subject to } \quad & Ax \geq b\\
        & \ \ x \in \{0,1\}^n
    \end{aligned}
\end{equation}
Requiring $x$ to be binary makes the feasible set discrete and hence nonconvex. Namely, we get
\begin{equation}\label{eq:feasible_set_binary}
    \mathcal{F} = \left \{ x \in \{0,1\}^n : Ax \geq b \right \}.
\end{equation}
The feasible set $\mathcal{F}$ is now finite, with at most $2^n$ feasible solutions. 
\begin{observation}\label{observation:infeasible_solution}
    It is important to note that in many problems, candidate solutions $z$ that aggressively minimize the objective $f(x)=c^{\mathsf T}x$ do so by violating one or more constraints in $ Axe \geq b$. Thus, the lowest-value solutions are generically expected to lie outside $\mathcal{F}$. In other words, it is often the case that there exist infeasible solutions $z$ such that
    \begin{equation*}
        f(z) \leq f(x) \text{ for all } x \in \mathcal{F}.
    \end{equation*}
    \hfill $\diamond$
\end{observation}

In general, solving 0--1 IP is NP-hard. Simplex or interior-point methods are insufficient for solving such problems. This creates opportunities for various heuristics and approximation algorithms, including quantum algorithms.

\subsection{Solving 0--1 Integer Programs}\label{sec:solving_binary_problems}
We briefly outline the classical strategy for solving 0--1 integer programs (IPs). As before, the simplex algorithm plays the central role in solving IP problems. It does so using \emph{LP relaxations} combined with two foundational ideas: branch-and-bound and cutting-plane methods \cite[Ch.~1.2]{conforti2014integer}.

Starting from the 0--1 feasible set $\mathcal{F}$, we first relax the integrality constraints and solve the LP relaxation
\begin{equation*}
    \mathcal{P}= \{ \, x \in \mathbb{R}^n : Ax \geq b,\ 0\leq x \leq 1\,\}.
\end{equation*}
Note that $\mathcal{F} \subseteq \mathcal{P}$. We then use simplex to obtain an optimal solution $x^\star \in \mathcal{P}$. If $x^\star$ is binary, it is also optimal for the original IP. Otherwise, $x^\star$ is fractional and further processing is required, namely, tightening the relaxed feasible set $\mathcal{P}$.

\begin{figure}[t]
    \centering
    \includegraphics[scale=0.068]{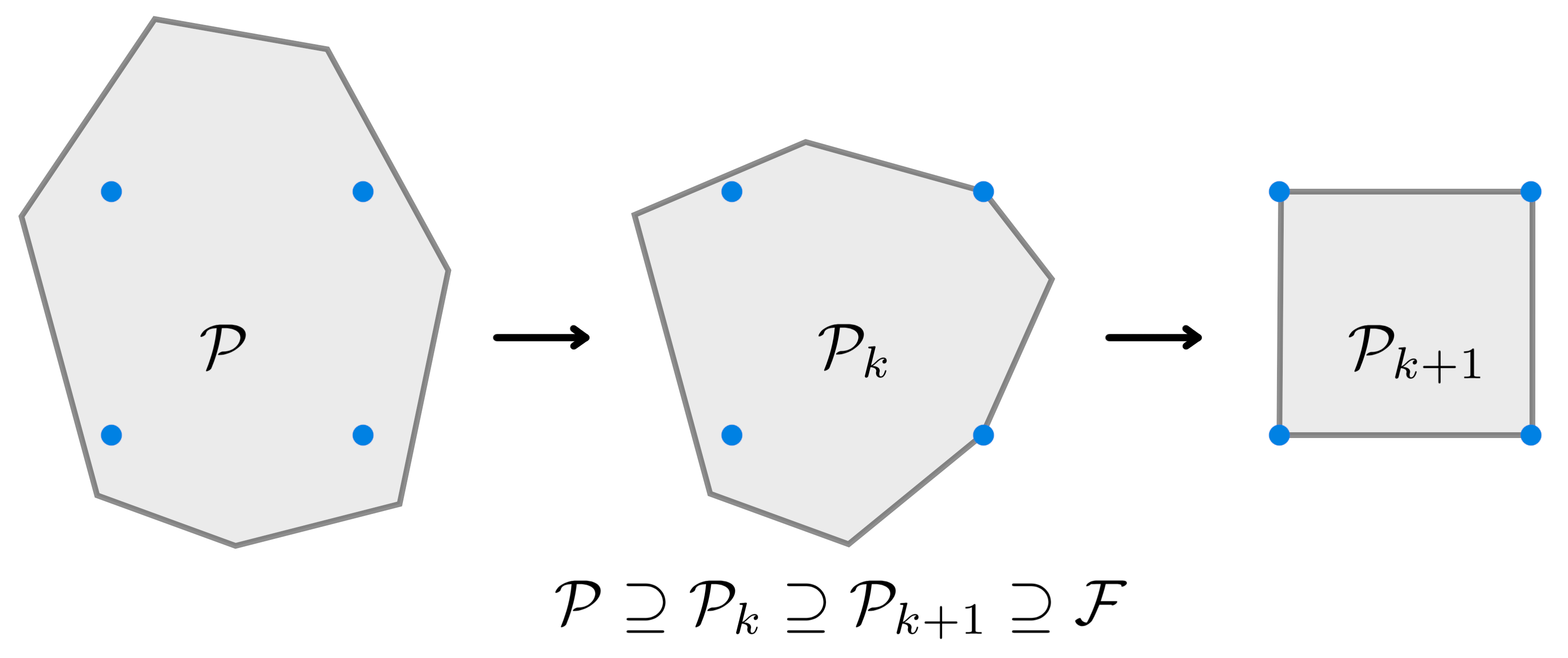}
    \caption{Conceptual illustration of iterative tightening in LP-based methods for solving 0--1 IPs. Starting from the LP relaxation $\mathcal{P}$, the goal is to generate a sequence of progressively tighter polyhedral relaxations $\mathcal{P}_k$ and $\mathcal{P}_{k+1}$ that continue to contain the discrete feasible set $\mathcal{F}$ (four blue points).}
    \label{fig:nested_sets}
\end{figure}

Both branch-and-bound and cutting-plane methods can be viewed as systematic ways to tighten the relaxation until integrality is enforced. Conceptually, they generate a sequence of relaxed search regions $\{\mathcal{P}_k\}$ that satisfy
\begin{equation}\label{eq:nested_relaxed_sets}
    \mathcal{P} \supseteq \mathcal{P}_k \supseteq \mathcal{P}_{k+1} \supseteq \mathcal{F},
\end{equation}
while repeatedly solving the associated LPs for each $\mathcal{P}_k$. The set nesting above is conceptual: cutting planes tighten the relaxation by ``slicing away" regions, while branch-and-bound tightens the search region by subdividing it into smaller relaxations and pruning regions that cannot contain an improving integer solution.

In this unified perspective, the algorithm alternates between (i) solving an LP relaxation to obtain bounds and candidate solutions and (ii) refining the feasible region by introducing additional restrictions. E.g., see \Cref{fig:nested_sets} diagram. They either subdivide the search region or add valid inequalities to progressively exclude fractional solutions until a binary optimum is eventually obtained.

In practice, these LP-based routines are often highly effective. However, because 0--1 IP is NP-hard in general, neither branch-and-bound nor cutting-plane methods admit polynomial-time guarantees: the number of refinements (nodes, cuts, or LP solves) grows exponentially in the worst case.

\subsection{0--1 Quadratic Unconstrained Problems with Penalties}
The IP formulation in \cref{eq:canonical_ip} is standard, and there are well-developed, powerful classical methods for solving it. By contrast, the rise of quantum optimization has popularized the QUBO formalism, in which constraints $Ax \geq b$ are absorbed into the objective via penalty terms, yielding an unconstrained binary problem. The appeal of this approach is that it maps directly to an Ising Hamiltonian, making it convenient for quantum implementations. However, penalty-based QUBO reformulations can introduce substantial difficulties relevant in both classical and quantum settings, as shown rigorously in \cite{gabbassov2025lagrangian}. Consequently, QUBO reformulations are best viewed as a hardware-driven modelling choice rather than the natural formulation for classical IP.

The IP in \cref{eq:canonical_ip} can be reformulated as a QUBO as follows. Define the linear objective function
\begin{equation}\label{eq:linear_objective}
    f(x) = c^{\mathsf T}x,
\end{equation}
and a quadratic function that penalizes constraint violations
\begin{equation}\label{eq:penalty_function}
    g(x,w) = \|Ax -b -w\|^2 = \sum_{i=1}^m \bigl((Ax-b)_i - w_i\bigr)^2.
\end{equation}
Here, $w \geq 0$ is a vector of slack variables. The definition and construction of the slack vector must ensure that for any feasible $x \in \mathcal{F}$, there exists $w$ such that $g(x,w)=0$. Conversely, if $x$ violates a constraint ($x \notin \mathcal{F}$), then $g(x,w)>0$ for all $w \geq 0$. The unconstrained problem is then
\begin{equation}\label{eq:unconstrained}
    \begin{aligned}
        &\min_{x \in \{0,1\}^n, w \geq 0}  \ q(x,w)\\
        &q(x, w) = f(x) + \lambda g(x,w),
    \end{aligned}
\end{equation}
where $\lambda > 0$ is the penalty multiplier that enforces feasibility. Note that we minimize over the binary vector $x$ and the non-binary slack variable $w \geq 0$. In the quantum computing context, the slack variable $w$ must be expressed as a function of additional binary variables $\gamma \in \{0,1\}^{\ell}$. Therefore, the final QUBO problem is
\begin{equation}\label{eq:min_qubo}
    \min_{x,\gamma} q(x,w(\gamma)),
\end{equation}
where $x \in \{0,1\}^n, \gamma \in \{0,1\}^\ell$.

% $N = n + \ell$ is the total number of binary variables.

\subsection{Lagrangian Methods}
Lagrangian multipliers and Lagrangian duality methods \cite{bertsekas2014constrained} are applicable to combinatorial optimization, but their roles differ from those in the smooth, continuous setting. Due to integrality constraints, Lagrangian duality is mainly used to obtain lower bounds, decompositions, and heuristic feasible solutions. The classical and quantum aspects of these approaches are studied in \cite{gabbassov2025lagrangian}.

\section{\label{sec:princip_quant_opt}Principles of Quantum Optimization with Constraints}
In this section, we introduce a time-dependent Hamiltonian whose evolution can realize the optimization process. In its most general form, it prescribes how to encode the optimization problem and how to explore the search space to find the optimal solution. This Hamiltonian provides a common starting point for many quantum optimization paradigms: QAOA-type algorithms \cite{wang2020xy,cook2020quantum,farhi2014quantum,herman2023constrained,bucher2025penalty,moll2018quantum,hadfield2019quantum}, Grover Search \cite{grover2001schrodinger,roland2002quantum}, Grover-extension D\"urr--H{\o}yer minimum-finding algorithm \cite{durr1996quantum} and quantum annealing \cite{rajak2023quantum}. The Hamiltonian is defined as:
\begin{equation}\label{eq:total_hamiltonian}
    H(t) = (1-s(t))\,H_{\text{init}} + s(t)\,H_{\text{p}}, \ \ 0 \le t \le T
\end{equation}
Here, $H_{\text{init}}$ and $H_{\text{p}}$ are the initial (mixer) and problem Hamiltonians, and the schedule $s(t)$ is a real monotonic function with $s(0)=0$ and $s(T)=1$. The mixer $H_{\text{init}}$ drives transitions between computational basis states (space exploration), while $H_{\text{p}}$ defines the energy landscape so that low-energy computational basis states correspond to good solutions.

The goal is to initialize the system in a suitable state and then evolve it under $H(t)$ so that, at the end of the evolution, the final state has a large overlap with the computational basis state encoding an optimal solution. Conventionally, the initial state is chosen to be a ground state of the mixer $H_{\text{init}}$. However, this choice is not strictly necessary. As we will see, relaxing this condition on the initial state in the presence of constraints leads to a more direct path to the optimal solution.

A central algorithmic choice is how to enforce constraints. The discussion below suggests that the computational difficulty of quantum encodings of integer-programming problems is driven primarily by the constraints, which generally induce many-body interactions and thereby increase the complexity of simulation and implementation. This is evident if one drops the constraints $Ax\ge b$ from \cref{eq:canonical_ip} and considers only the unconstrained minimization of $f(x)=c^{\mathsf T}x$ over $x \in \{0,1\}^n$. The resulting problem is separable across variables: each component $x_i$ can be chosen independently according to the sign of $c_i$, so the problem reduces to $n$ independent one-variable minimizations. The nontrivial difficulty, therefore, lies in enforcing feasibility, which, in the quantum setting, is reflected in the need to encode these constraints through nontrivial many-qubit interactions.

\subsection{Penalty-Based Methods}
In penalty-based approaches, $H_{\text{p}}$ encodes the QUBO function $q(x,w(y))$ from \cref{eq:min_qubo}. Since this function is quadratic, it yields a 2-local Ising Hamiltonian with nontrivial qubit interactions. A standard choice of mixer is the 1-local transverse-field Hamiltonian
\begin{equation}\label{eq:transverse_field}
H_{\text{init}} = -\sum_{i=1}^N X_i,
\end{equation}
where $N = n + \ell$, and $X_i$ is the Pauli $X$ operator acting on qubit $i$. Because $H_{\text{init}}$ is 1-local, the total Hamiltonian $H(t)$ remains 2-local. 

Let $r = (x,\gamma) \in \{0,1\}^N$ with $x \in \{0,1\}^n$ representing the problem variables and $\gamma \in \{0,1\}^\ell$ encoding the slack variable $w$ [see \cref{eq:min_qubo}].
Then the mixer in \cref{eq:transverse_field} induces transitions between computational basis states that differ by a single bit flip, namely from $\ket{r}$ to $\ket{r\oplus e_i}$, where $e_i$ is the bit string with a 1 in position $i$ and 0's elsewhere. Since these transitions are defined independently of the constraints, the mixer does not preserve feasibility and can drive the dynamics into the infeasible subspace. Feasibility is enforced instead by $H_{\text{p}}$, which assigns energetic penalties to infeasible solutions.

This mixer admits a graph-theoretical interpretation. The negative of $H_{\text{init}}$ is the adjacency matrix of a $N$-dimensional hypercube with vertex set $V = \{0,1\}^{N}$ and the edge set $E=\{(r,r\oplus e_i): r \in \{0,1\}^N,\ i=1,\dots,N\}$. Thus, each vertex represents a feasible or infeasible solution, and each edge represents a single-bit-flip transition allowed by the mixer.

\begin{figure}[t]
    \centering
    \includegraphics[scale=0.11]{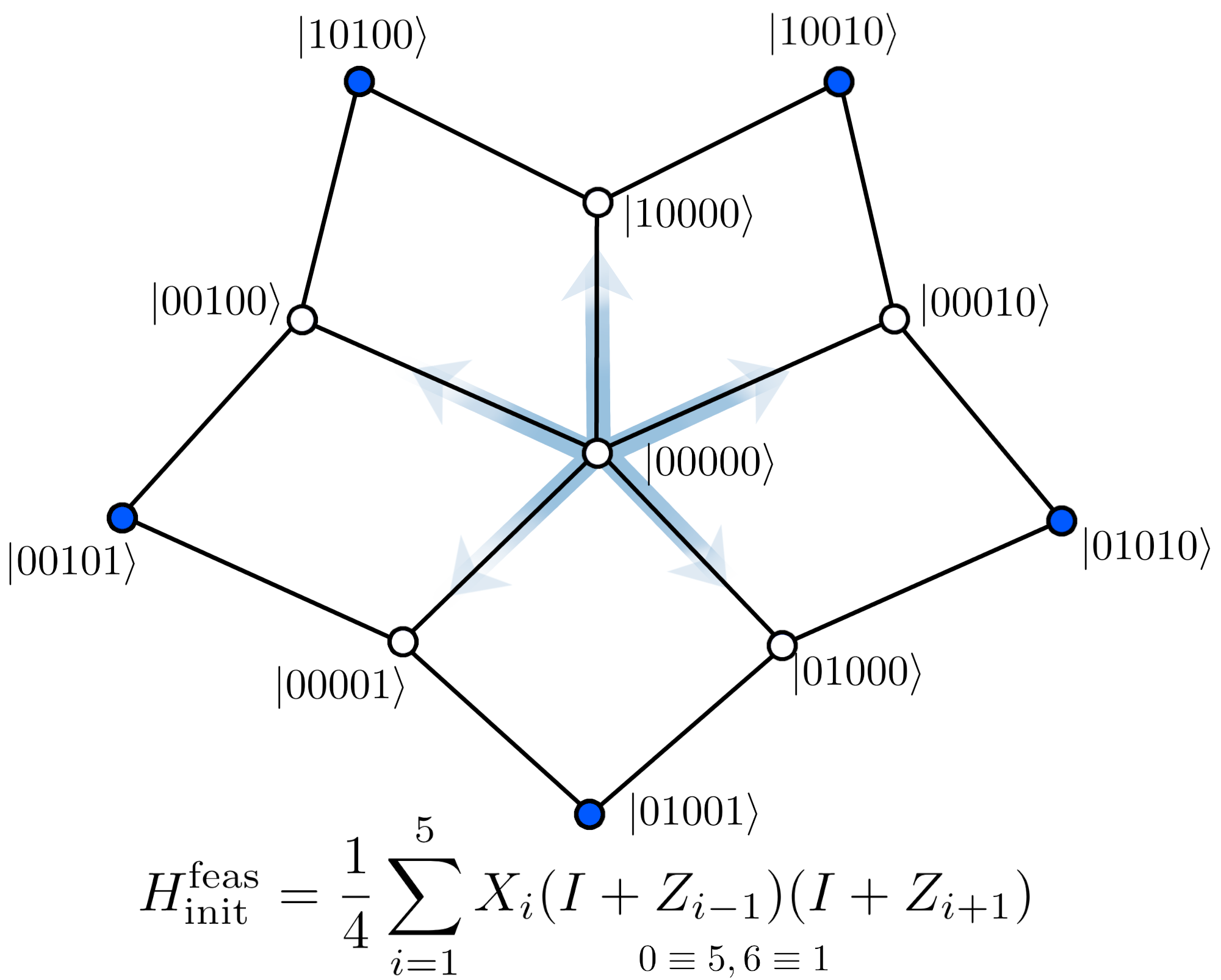}
    \caption{Transition graph induced by a feasibility-preserving mixer for the Maximum Independent Set Problem on the cycle graph $C_5$, [see \Cref{ap:maximum_independent_set}]. Each vertex is a feasible bit string $\ket{x}$ encoding an independent set, and each edge represents an allowed transition that preserves feasibility. The transparent blue arrows illustrate wavefunction propagation in superposition from $\ket{00000}$ to adjacent feasible states. The outer blue vertices are the maximum independent sets of $C_5$ and therefore the optimal solutions.}
    \label{fig:feas_mixer_max_indep_set}
\end{figure}

\subsection{Penalty-Free Methods}
In penalty-free approaches, $H_{\text{p}}$ typically encodes only the linear objective $f(x) =c^{\mathsf{T}}x$ in \cref{eq:canonical_ip}, making $H_{\text{p}}$ 1-local and efficiently simulatable. However, the difficulty is shifted into $H_{\text{init}}$, which must enforce feasibility and is generally high-$k$-local. A generic feasible-subspace mixer can be written as
\begin{equation}\label{eq:feasible_subspace_mixer}
    H_{\text{init}}^{\text{feas}}
    = \sum_{ x,x' \in \mathcal{F}}
    M_{x'x}\,|x'\rangle\langle x|, \quad M = M^{\dagger},
\end{equation}
where $M_{x'x} \in \mathbb{C}$ are chosen such that the eigenvalue $0$ is not the lowest eigenvalue, and $\mathcal{F}$ is the feasible set defined in \cref{eq:feasible_set_binary}.  This mixer connects feasible states and does not allow transitions to infeasible subspaces because
\begin{equation}\label{eq:mixer_kernel}
     H_{\text{init}}^{\text{feas}} \ket{z} = 0 \text{ for } z \notin \mathcal{F}.
\end{equation}
Since $H_{\text{p}}$ is diagonal and $H_{\text{init}}^{\text{feas}}$ has support only on feasible states, the Hilbert space decomposes into two invariant subspaces,
\begin{equation}
\mathcal{H} = \mathcal{H}_{\mathcal{F}} \oplus \mathcal{H}_{\mathcal{Q}},
\end{equation}
where 
$$\mathcal{H}_{\mathcal{F}}=\operatorname{span} \left \{ \ket{x}: x \in \mathcal{F} \right \}, \quad  \mathcal{H}_{\mathcal{Q}}=\operatorname{span} \left \{ \ket{z}: z \notin \mathcal{F} \right \}.$$

As before, the feasibility-preserving mixer admits a graph-theoretic interpretation. If $M_{x'x} \in \mathbb{R}$, then $ H_{\text{init}}^{\text{feas}}$ is proportional to the adjacency matrix of an undirected weighted graph on vertex set $V = \mathcal{F}$, and the edge set  $E = \{(x,x'):M_{x'x}\neq 0\}$. Each edge corresponds to an allowed transition between $|x\rangle$ and $|x'\rangle$, and transition rates are controlled by $|M_{x'x}|$.

One might wonder how a feasibility-preserving mixer can be realized. For certain constraint types, such mixers can be constructed explicitly using polynomially many Pauli operators \cite{hadfield2019quantum, wang2020xy, cook2020quantum}. More recent and general techniques realize feasibility-preserving mixers indirectly by suppressing infeasible state transitions via repeated projective measurement schemes \cite{herman2023constrained,pawlak2023quantum,fuchs2024lx}. Other methods, such as \cite{bucher2025penalty, bucher2025efficient}, use oracle-based subroutines to allow only feasible transitions.

\begin{observation}
    In the classical LP setting, the simplex method moves along the edges of the feasible polyhedron, where vertices are feasible solutions and edges are allowable transitions, as shown in \Cref{fig:polyhedron_and_graph} (b). Analogously, $H_{\text{init}}^{\text{feas}}$ induces a transition graph \mbox{$G=(V,E)$}.
    Its vertices are feasible computational basis states, and its edges are the transitions allowed by the mixer, as illustrated in \Cref{fig:feas_mixer_max_indep_set}. The difference is that the quantum counterpart coherently explores the allowed transitions in superposition. \hfill $\diamond$
\end{observation}

\subsection{General Mixer}
Restricting attention to either unconstrained mixers, such as \cref{eq:transverse_field}, or the feasible-subspace mixer in \cref{eq:feasible_subspace_mixer}, is unnecessarily limiting. To unify these cases, as well as hybrids of them, we introduce the orthogonal projectors $\Pi_{\mathcal{F}}$ and $\Pi_{\mathcal{Q}}$. The action of $\Pi_{\mathcal{F}}$ is
\begin{equation}\label{eq:projector_F}
    \Pi_{\mathcal{F}}\ket{x} =
    \begin{cases}
        \ket{x}, & x \in \mathcal{F},\\
        0, & \text{otherwise},
    \end{cases}
\end{equation}
and $\Pi_{\mathcal{Q}} := I - \Pi_{\mathcal{F}}$. Thus, $\Pi_{\mathcal{F}}$ projects onto the feasible subspace, while $\Pi_{\mathcal{Q}}$ projects onto its orthogonal complement. Then, a general mixer can be written as
\begin{equation}\label{eq:relaxed_mixer}
    H_{\text{init}}^{\epsilon} = (\Pi_{\mathcal{F}} + \epsilon \Pi_{\mathcal{Q}})H_{\text{init}}(\Pi_{\mathcal{F}} + \epsilon \Pi_{\mathcal{Q}}),
\end{equation}
where $H_{\text{init}}$ is an unconstrained mixer, e.g., \cref{eq:transverse_field} and $\epsilon >0$ controls how much leakage is allowed into the infeasible subspace. Leakage can be intentionally allowed to better explore the instantaneous landscape or used to model undesirable transitions arising from implementation errors in a feasible-subspace mixer. For $\epsilon = 0$, this expression reduces to the feasible-subspace mixer, e.g., \cref{eq:feasible_subspace_mixer}. For $\epsilon = 1$, since $\Pi_{\mathcal{F}} + \Pi_{\mathcal{Q}} = I$, it recovers the unconstrained mixer $H_{\text{init}}$.

Alternatively, we can rewrite the mixer in a block form
\begin{equation}\label{eq:relaxed_mixer_extended}
    H_{\text{init}}^{\epsilon} = H^{\mathrm{init}}_{FF} + \epsilon H^{\mathrm{init}}_{FQ} + \epsilon H^{\mathrm{init}}_{QF} + \epsilon^2 H^{\mathrm{init}}_{QQ},
\end{equation}
where $H^{\mathrm{init}}_{AB} \coloneq \Pi_{\mathcal{A}}H_{\text{init}}\Pi_{\mathcal{B}}$.  The first term in \cref{eq:relaxed_mixer_extended} generates transitions entirely within the feasible subspace. The second and third terms generate $ \epsilon$-suppressed transitions between feasible and infeasible states, while the final term generates rare $\epsilon^2$-suppressed transitions within the infeasible subspace.

\section{Entanglement Restructuring and Computational Slowdown}\label{sec:entanglement_restructuring_and_slowdown}
Many quantum optimization algorithms \cite{farhiquantum,hadfield2019quantum,wang2020xy,cook2020quantum,fuchs2022constraint,hao2026constraint,bartschi2020grover,bucher2025penalty,herman2023constrained,pawlak2023quantum,feinstein2025robustness}, whether analog or gate-based, ultimately arise from Hamiltonian dynamics. The spectral properties of the relevant Hamiltonians therefore play a central role in the performance of many quantum algorithms. In adiabatic quantum computation, and more generally in regimes where the dynamics is organized around avoided level crossings, the minimum spectral gap is traditionally regarded as the quantity that determines the minimum required evolution time and hence limits the achievable computational speed \cite{Farhi2000AdiabaticEvolution,AlbashLidar2018AQC}. In gate-based quantum computation, which is polynomially equivalent to adiabatic quantum computation \cite{aharonov2008adiabatic}, the spectral properties manifest as the complexity and depth of quantum circuits that approximate the Hamiltonian evolution \cite{mcdowall2026spectral,nzongani2026scaling,deshpande2022importance,dooley2020simulating,lin2020near,roca2023circuit,vzunkovivc2026scalable,wang2025imaginary,hopkins2025multi}. Therefore, the spectral perspective developed in this section is relevant to a broad family of quantum algorithms.

In \cite{gabbassov2025adiabatic}, it has been rigorously shown that the dramatic computational slowdown occurs when a quantum system, while following an instantaneous eigenstate, needs to rapidly restructure its entanglement relations among subsystems, e.g., creating, destroying, or redistributing entanglement among qubits over a short interval of Hamiltonian interpolation. This rapid restructuring of entanglement can only occur at narrowly avoided level crossings. Consequently, in \cite{gabbassov2025adiabatic}, it was demonstrated that neither the minimum gap nor the amount of entanglement is an actual cause of a computational slowdown, but rather the extent and speed of entanglement restructuring. The minimum spectral gap is merely a consequence of the underlying physical changes in entanglement. This holds true for both adiabatic and gate-based algorithms due to their aforementioned equivalence. In the gate-based model, the same intuition appears in a different language: a substantial reorganization of entanglement across distant subsystems requires a sufficiently deep circuit with many entangling gates and long sequences of SWAP gates to mediate the transfer of entanglement.

The relationship between entanglement restructuring and narrow avoided-level crossings suggests a route to faster evolution: avoid entanglement restructuring that is not needed to reach the target state.

\section{Case Study: Entanglement Restructuring for Quantum Algorithm Design}\label{sec:case_study}
In this section, we demonstrate the importance of minimizing entanglement restructuring in the design of efficient quantum optimization algorithms. We also show that, regardless of the chosen method (penalty-free or penalty-based) to achieve faster computation, one needs to be aware of the algebraic and contextual structure of an optimization problem. Applying one method or the other blindly will almost certainly result in an inefficient algorithm at best and a catastrophically failing one at worst. Therefore, there is no \textit{free lunch}: no single algorithmic strategy can be expected to perform well across all problem classes without exploiting the particular structure of the instance at hand \cite{wolpert2002no}.

To demonstrate these points, let us consider a trivial yet very ubiquitous 0--1 IP example. Consider an optimization problem:
\begin{equation}\label{eq:importrant_example}
    \begin{aligned}
        \text{min }  \ & c^{\mathsf{T}}x \\
        \text{subject to } \ & x_3 + x_4 = 1 \\
        & x \in \{0,1\}^4
    \end{aligned}
\end{equation}
Here, the vector $c$ has positive non-zero entries such that the optimal solution is 
\begin{equation}\label{eq:optimal_solution}
    x^\star = (0010).
\end{equation}
Since $c >0$, to minimize $f(x)$, we want to switch off as many bits as possible while satisfying the constraint $x_3 + x_4 = 1$. Assuming $c_3 < c_4$, it is easy to see that \cref{eq:optimal_solution} is the optimal solution. The constraint 
$$x_3 + x_4 =1$$
is a particular form of the cardinality constraint, and it is one of the most common and versatile constraints in optimization \cite{williams2013model}. This constraint represents the choice: out of two variables $x_i$ and $x_j$, exactly one must be selected (activated). The choice of one item over the other arises naturally in many real-life decision processes and, hence, in many optimization contexts. We will show that such an innocuous and simple constraint induces highly nontrivial quantum dynamics. In this dynamics, the importance of entanglement restructuring and spectral duality becomes especially apparent.

\subsection{Minimizing Entanglement Restructuring}
Here, we relate computational slowdown to entanglement restructuring and show how this relation can guide the design of faster quantum optimization algorithms. We first consider a penalty-based method, where large penalties induce abrupt entanglement restructuring and hence narrow the minimum spectral gap. We then consider a penalty-free method with a feasibility-preserving mixer that completely eliminates entanglement restructuring and yields fast evolution. However, this method can also lead to a catastrophic failure in which the optimal solution is never sampled, regardless of the available computational resources. Finally, we show that a slightly richer feasibility-preserving mixer avoids this failure while keeping entanglement restructuring minimal.

\subsubsection{Penalty-based setup}
Let us consider the penalty-based approach first. The objective function is quadratic, and it is given by:
\begin{equation}
    q(x) = \sum_{i=1}^4 c_i x_i + \lambda (x_3 + x_4 - 1)^2
\end{equation}
The infeasible solutions of the form $(x_1, x_2, 1, 1)$ and $(x_1, x_2, 0, 0)$ are penalized by the energetic penalty $\lambda$, and the solutions of the form $(x_1, x_2, 1, 0)$ and $(x_1, x_2, 0, 1)$ have no penalty.

Encoding the objective function yields a 2-local problem Hamiltonian
\begin{equation}
    H_{\mathrm{q}} = -\frac{1}{2}\sum_{i=1}^4 c_i Z_i +\frac{\lambda}{2} (I + Z_3 Z_4).
\end{equation}
Here, we omitted additive identity terms since they do not affect dynamics in any way. Therefore, we identify eigenvalues with the objective function $f(x)$ or $q(x)$ up to this shift; for Hamiltonian encoding details, see \Cref{ap:hamiltonian_encoding}.

For the initial Hamiltonian, we choose a 1-local transverse-field mixer
\begin{equation}\label{eq:4_qubit_transverse_field}
    H_{\mathrm{init}} = -\sum_{i=1}^4 X_i,
\end{equation}
and initialize the system in its unique ground state
\begin{equation*}
    \ket{\psi(0)} = \frac{1}{\sqrt{16}}\sum_{x \in \{0,1\}^4}\ket{x} = \ket{+}^{\otimes 4}.
\end{equation*}
This state is a uniform superposition of all computational basis states and hence a product state.

A large $\lambda > 0$ enforces feasibility but also induces significant entanglement restructuring near the narrowing spectral gap. This behaviour is shown by the solid red curves in \Cref{fig:penalty_vs_free_entropy_and_gap}. In the top panel, the solid red curves show the bipartite entanglement entropy of the reduced system's state as it evolves. As $\lambda$ increases, the entropy peak becomes steeper and narrower, indicating that the entanglement structure changes over a shorter interval of the schedule $s(t)$. In the middle panel, the corresponding solid red curves show the absolute rate of change of entropy, making this abrupt restructuring explicit. In the bottom panel, the solid red curves show the associated narrowing of the spectral gap, which is a consequence of this rapid entanglement restructuring. This suggests that the computational slowdown is not determined by the peak height, which measures the maximum amount of entanglement, but by how rapidly the entanglement structure changes. \Cref{fig:3_methods_transitions} (a) provides the complementary state transition graph picture: infeasible solutions form energetic penalty barriers, and increasing these barriers makes it harder for amplitudes, initially spread over all computational basis states, to reconcentrate at the optimal feasible solution marked in blue colour.

\begin{figure}[t]
    \centering
    \includegraphics[scale=.11]{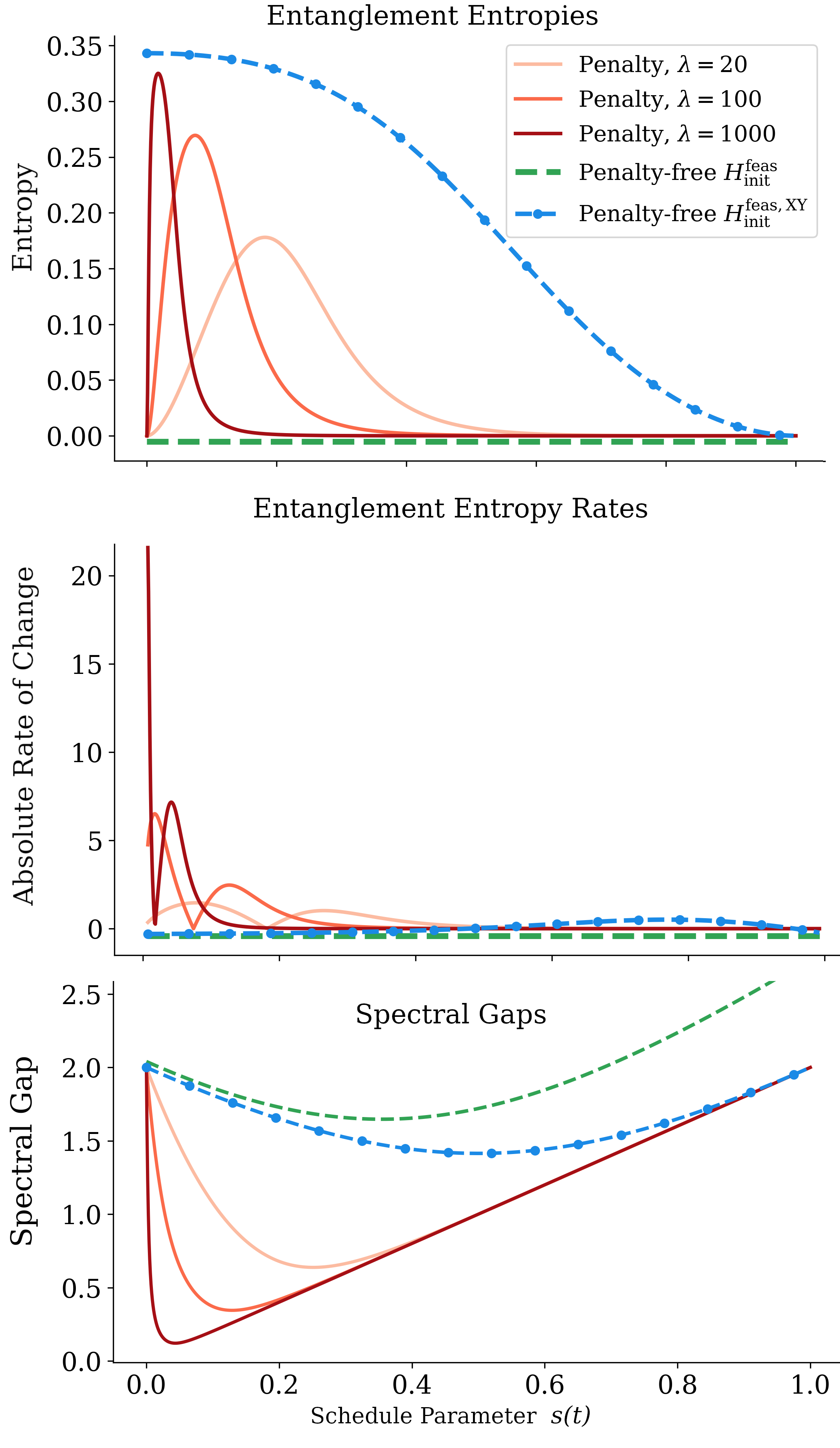}
    \caption{Bipartite entanglement entropy of the reduced system state (top), the absolute rate of change of the entropy (mid), and spectral gaps (bottom) for penalty-based and penalty-free methods. For the penalty-based method, increasing the penalty strength makes entanglement restructuring more abrupt, producing narrower peaks with steeper slopes and hence narrowing the spectral gap. For the penalty-free methods (with mixers $H^{\mathrm{feas}}_{\mathrm{init}}$ and $H_{\mathrm{init}}^{\mathrm{feas, XY}}$), the entropy rates are either zero or close to zero, showing that entanglement restructuring is either removed or minimal; correspondingly, the relevant spectral gap remains wide throughout the evolution.}
\label{fig:penalty_vs_free_entropy_and_gap}
\end{figure}

\subsubsection{Penalty-free setup}
It is possible to eliminate entanglement restructuring entirely, thereby achieving a significantly faster evolution time. For this, we consider the penalty-free approach. Then, the problem Hamiltonian is 1-local,
\begin{equation}
    H_{\mathrm{f}} = -\frac{1}{2}\sum_{i=1}^4 c_i Z_i.
\end{equation}
First, consider the simplest feasibility-preserving mixer, obtained from the transverse-field Hamiltonian in \cref{eq:4_qubit_transverse_field} by excluding infeasible transitions:
\begin{equation}\label{eq:4_qubit_trans_feas}
    H^{\mathrm{feas}}_{\mathrm{init}} = - X_1 - X_2.
\end{equation}
Hence, the total Hamiltonian is 1-local. Due to 1-locality, the system neither creates, redistributes, nor destroys entanglement. This is shown by the dashed green curves in \Cref{fig:penalty_vs_free_entropy_and_gap}: both the entanglement entropy and its rate of change remain zero throughout the evolution, while the relevant spectral gap remains wide. Thus, by removing entanglement restructuring, the evolution becomes trivial and fast. The corresponding state transition graph is shown in \Cref{fig:3_methods_transitions} (b), where the mixer preserves feasibility by allowing transitions only within feasible components.

\begin{figure}[t]
    \centering
    \includegraphics[scale=.14]{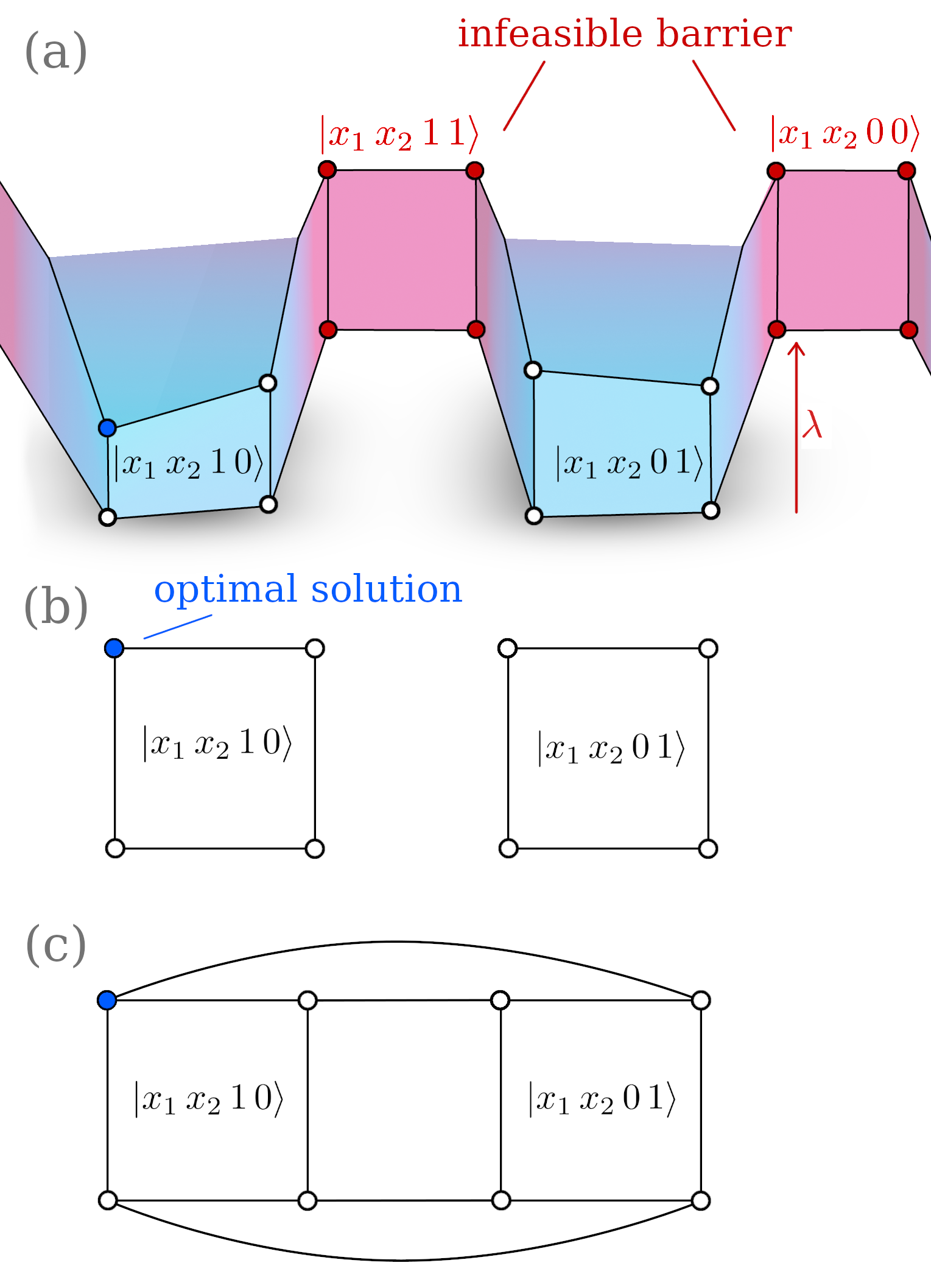}
    \caption{Transition graphs for three quantum optimization methods. (a) Penalty-based method with its transition graph (4-d hypercube) overlaid over the energy landscape. Vertices with higher energy are raised upward. Infeasible vertices are raised by the penalty $\lambda$, creating an energetic barrier between feasible regions. The initial state has support on all vertices, and reaching the optimal solution requires some amplitudes to traverse the barriers and concentrate at the optimal solution (blue vertex). For visual clarity, only representative transitions are shown. (b) Transition graph for the penalty-free method with $H^{\mathrm{feas}}_{\mathrm{init}}=-X_1-X_2$. The energy landscape is trivial and given by a linear function $f(x)$. The transition graph splits into two disconnected feasible components, so the initial ground space is degenerate, and ground states are supported on different components.  If the system is initialized in the wrong component, the optimal solution is dynamically inaccessible. (c) Transition graph of the penalty-free method with $H^{\mathrm{feas,XY}}_{\mathrm{init}}$. The additional $XY$ term introduces transitions $\ket{x_1x_2 01}\leftrightarrow\ket{x_1x_2 10}$, connecting the feasible components and making the optimal solution dynamically accessible while minimizing entanglement restructuring.}
\label{fig:3_methods_transitions}
\end{figure}

However, this approach is prone to a catastrophic failure of never sampling the optimal solution. To see this, we note that the full ground space of the initial feasibility-preserving Hamiltonian $H_{\text{init}}^{\text{feas}}=-X_1-X_2$ is degenerate, since the Hamiltonian acts trivially on qubits 3 and 4. Its intersection with the feasible subspace is two-dimensional and is spanned by feasible ground states,
\begin{equation*}
    \ket{u_{10}} = \ket{+}\ket{+}\ket{1}\ket{0}, \quad \ket{u_{01}} = \ket{+}\ket{+}\ket{0}\ket{1}.
\end{equation*}
Both ground states belong to distinct orthogonal feasible subspaces that are invariant under the dynamics of the total Hamiltonian:
\begin{equation*}
     \ket{u_{10}} \in \mathcal{H}_{10}, \quad \ket{u_{01}} \in \mathcal{H}_{01}
\end{equation*}
Since $x^\star=(0010)$, we have $\ket{x^\star}\in\mathcal{H}_{10}$. Hence, if the system is initialized in $\ket{u_{01}}$, then invariance of $\mathcal{H}_{01}$ prevents any amplitude from reaching $\ket{x^\star}$. The optimal solution, therefore, has zero sampling probability for all evolution times. \Cref{fig:3_methods_transitions} (b) shows that the state transition graph is a disconnected two-component graph. Each component spans orthogonal subspaces $ \mathcal{H}_{10}$ and $ \mathcal{H}_{01}$.

The mixer in \cref{eq:4_qubit_trans_feas} shows that completely eliminating entanglement restructuring can make the evolution fast, but can also lead to catastrophic failure. This failure can be avoided with a different feasibility-preserving mixer. The new mixer introduces only minimal entanglement restructuring, so the speedup is maintained.
\begin{equation}\label{eq:case_study_mixer}
   H_{\mathrm{init}}^{\mathrm{feas, XY}} = -X_1 - X_2 -\frac{1}{2}(X_3 X_4 + Y_3 Y_4).
\end{equation}
Just as the previous mixer, $H_{\mathrm{init}}^{\mathrm{feas, XY}}$ keeps qubits 1 and 2 unentangled. The 2-local terms introduce feasible transitions between previously invariant subspaces $\mathcal{H}_{01}$ and $\mathcal{H}_{10}$. As a result, the degeneracy between the two feasible components is lifted, and the ground state is supported entirely on feasible computational basis states
\begin{equation}\label{eq:ground_state_of_feas_xy}
    \ket{u} = \ket{+}\ket{+}\frac{\ket{01}+\ket{10}}{\sqrt{2}}.
\end{equation}
\Cref{fig:penalty_vs_free_entropy_and_gap} shows that entropy starts off high and slowly decays to zero (dashed dotted blue curve). The entropy rate of change remains small, showing that the entanglement is not rapidly restructured. Consequently, the relevant spectral gap remains wide. We note that the entropy peak is tall because the initial state $\ket{u}$ in \cref{eq:ground_state_of_feas_xy} has qubits 3 and 4 maximally entangled. This, however, as we noted before, does not affect the slowdown, since the spectral gap remains unchanged. The corresponding transition graph is shown in \Cref{fig:3_methods_transitions} (c): the additional $XY$ term connects the previously disconnected feasible components by allowing transitions of the form \mbox{$\ket{x_1 x_2 0 1} \leftrightarrow \ket{x_1 x_2 1 0}$}.

This example illustrates how entanglement restructuring informs algorithm design. The objective is not necessarily to eliminate entanglement restructuring altogether, but to keep it minimal while ensuring that the optimal solution remains dynamically accessible. We now turn from this example to the general theory, which identifies entanglement restructuring as the central quantity governing the computational difficulty of constrained quantum optimization.

\section{Theory of Entanglement Dynamics}\label{sec:entanglement_dynamics}
The discussion in this section is not specific to quantum optimization. Here, we develop a general theory explaining the mechanics of entanglement dynamics and restructuring. Later, we apply this theory to the analysis of quantum optimization.

In the spectral description, entanglement restructuring manifests as narrow avoided level crossings. In \cite{gabbassov2025adiabatic}, it was shown that the more rapid and abrupt the restructuring is, in the sense that a substantial change in the entanglement structure is compressed into a short interval $\Delta s$ of the interpolation schedule $s=s(t)$, the narrower the corresponding avoided level crossing must be. Therefore, to carry out substantial entanglement restructuring while remaining in an instantaneous eigenstate, the evolution must slow down considerably near the avoided crossing. Intuitively, when two eigenvalues of a time-dependent Hamiltonian closely approach and repel each other, they rapidly exchange their eigenvectors. Thus, if the system occupies one of these instantaneous eigenstates, then near the avoided crossing the system's state is exchanged with the eigenvector associated with the avoided eigenvalue; see \Cref{fig:eigenvector_swap_schematic}. Typically, the exchanged eigenvector carries a different entanglement structure.

\subsection{Analytical Example of Entanglement Restructuring}
We now illustrate how an eigenvector swap can force entanglement restructuring. Consider a $4$-level initial Hamiltonian written in its eigenbasis,
\begin{equation*}
    H_{\text{init}} = \sum_{k=1}^4 d_k \ket{d_k}\bra{d_k},
\end{equation*}
where the eigenvalues $d_k$ are distinct and the eigenvectors $\ket{d_k}$ are normalized. Assume that $\ket{d_1}$ is a product state, while $\ket{d_2}$ is maximally entangled. We introduce a rank-one interaction Hamiltonian,
\begin{equation*}
    H_v(t) = H_{\text{init}} + \mu(t)\ket{v}\bra{v},
\end{equation*}
where $\ket{v}$ is normalized and the schedule $\mu(t)$ is monotonically increasing, with $\mu(0)=0$ and $\mu(t)>0$ for $t>0$.

\begin{figure}[t]
    \centering
    \includegraphics[scale=0.13]{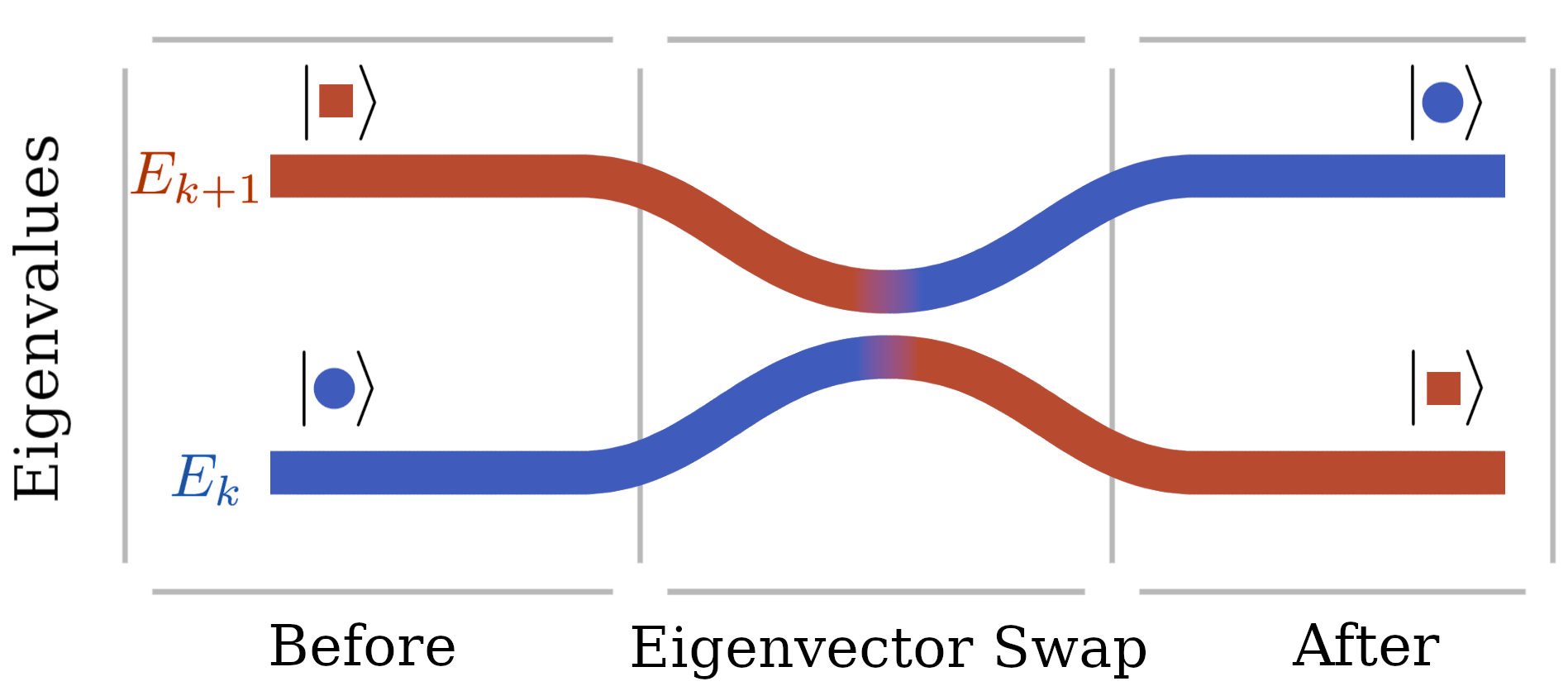}
    \caption{Schematic illustration of an eigenvector swap at a narrowly avoided level crossing. The vertical axis represents eigenvalues, and the horizontal axis represents time. Before the avoided crossing, the eigenvalue trajectory $E_k$ is associated with the blue-circle eigenvector, while $E_{k+1}$ is associated with the brown-square eigenvector. The colour of each eigenvalue trajectory indicates the eigenvector associated with that trajectory. At the avoided level crossing, the two eigenvalues rapidly exchange their eigenvectors (the upper trajectory changes from brown to blue, while the lower trajectory changes from blue to brown). This rapid exchange of eigenvectors is the spectral mechanism underlying entanglement restructuring \cite{gabbassov2025adiabatic}.}
    \label{fig:eigenvector_swap_schematic}
\end{figure}

Suppose the system is initialized in the ground eigenstate $\ket{d_1}$ and is required to remain in the instantaneous ground eigenstate $\ket{E_1(\mu)}$ of $H_v(t)$. That is
\begin{equation*}
    \ket{\psi(0)} = \ket{E_1(0)} = \ket{d_1}, \text{ and } \ket{\psi(\mu)} = \ket{E_1(\mu)}.
\end{equation*}
The largest possible entanglement restructuring in this example occurs if the state evolves from the product state $\ket{d_1}$ to the maximally entangled state $\ket{d_2}$. To see how this can occur, first choose $\ket{v}$ such that
$$|\langle v |d_1 \rangle|^2 = 1.$$
Then the eigenvalue associated with $\ket{d_1}$ becomes $d_1+\mu$, while all other eigenvalues remain fixed. Hence, as $\mu$ increases, this eigenvalue crosses the higher eigenvalues $d_2,d_3,d_4$ exactly, and the corresponding instantaneous eigenvector remains $\ket{d_1}$. Now choose a different $\ket{v}$ so that, for $\epsilon \ll 1$,
\begin{equation*}
    |\langle v| d_1 \rangle|^2 = 1-O\left (\epsilon^2 \right ) \ \text{ and } \ |\langle v |d_j \rangle|^2 = O\left (\epsilon^2 \right ).
\end{equation*}
Order $\epsilon$ components of $\ket{v}$ along eigenvectors $\ket{d_j}$ for $j \neq 1$, turn the exact crossings into narrow $\epsilon$-controlled avoided level crossings; see \Cref{fig:edge_case_cascade}.

In particular, the first such avoided level crossing occurs between the instantaneous eigenvalues $E_1(\mu)$ and $E_2(\mu)$ at some critical value $\mu_{\text{critical}}>0$. At this avoided level crossing, the two eigenvalues rapidly exchange their eigenvectors \cite{gabbassov2025adiabatic}. The exchange is rapid because it occurs over a short interval proportional to $\Delta \mu$. Hence, if the system remains on the instantaneous ground energy level, then after the crossing, its state becomes approximately
\begin{equation}\label{eq:psi_d2}
    \ket{\psi(\mu_{\text{critical}}+\Delta\mu)} \approx \ket{d_2}.
\end{equation}
In the limit $\epsilon\to0$, this approximation becomes exact, while the gap at the avoided level crossing is arbitrarily close to zero. Thus, the state rapidly changes from a product state to a maximally entangled state through an eigenvector swap.

This example shows the trade-off directly. Near the avoided level crossing, the instantaneous ground eigenvector changes from approximately $\ket{d_1}$ to approximately $\ket{d_2}$ over a small interval of $\mu$; hence, by the adiabatic theorem, the evolution must slow down near the crossing if the system's state is to remain aligned with the instantaneous ground eigenstate. Thus, smaller $\epsilon$ produces a narrower avoided crossing and a sharper entanglement restructuring, but also requires slower evolution to remain on the same instantaneous energy level.

Since such avoided crossings will play a central role in the following sections, we give them a special name.
\begin{definition}
    For $\epsilon > 0$, the narrowly avoided level crossings whose gap vanishes (closes) with $\epsilon$ and is bounded below at the order $\epsilon^2$ are called
    \textit{\mbox{$\epsilon$-avoided} level crossings}.
    \hfill $\diamond$
\end{definition}

\subsection{Theorems of Entanglement Restructuring}
We will formalize the example above in the following theorem. The theorem establishes a sufficient condition for an eigenvector swap that is necessary for rapid entanglement restructuring.

\begin{theorem}[Eigenvector Swap \cite{gabbassov2025adiabatic}]\label{theorem:eigenvector_swap}
    Let $D_{\mathrm{init}}$ be a Hamiltonian with eigenvalues
    $$
        d_1 < d_2 < \cdots < d_N,
    $$
    and corresponding normalized eigenvectors $\ket{d_j}$ for $j=1,\dots,N$. Let $\ket{v}$ be normalized, and let $\mu(t)\in[0,\infty)$ be a monotonically increasing schedule function. Define the time-dependent Hamiltonian
    \begin{equation}\label{eq:edge_case_hamiltonian}
        H_v(t) = D_{\mathrm{init}}+\mu(t)\ket{v}\!\bra{v}.
    \end{equation}
    Suppose that for $\epsilon > 0$ and some fixed $k$ and all $j \neq k$,
    \begin{equation}\label{eq:swap_condition}
        |\langle v| d_k \rangle|^2 = 1-O(\epsilon^2) \ \text{ and } \ |\langle v |d_j \rangle|^2 = O\left (\epsilon^2 \right).
    \end{equation}
    Then, in the limit $\epsilon\to0$, the instantaneous eigenvalues
    $E_k(t)$ and $E_{k+1}(t)$ of $H_v(t)$ exchange their eigenvectors at
    the \mbox{$\epsilon$-avoided} level crossing occurring at some \mbox{$\mu(t') > 0$}. That is, for $\Delta \mu > 0$ sufficiently small, and $\epsilon \to 0$, before the \mbox{$\epsilon$-avoided} level crossing, one has
    \begin{equation*}
        \ket{E_k(\mu(t')-\Delta\mu)}=\ket{v}, \quad \ket{E_{k+1}(\mu(t')-\Delta\mu)}=\ket{d_{k+1}},
    \end{equation*}
    whereas after the \mbox{$\epsilon$-avoided} level crossing, one has
    \begin{equation*}
        \ket{E_k(\mu(t')+\Delta\mu)}=\ket{d_{k+1}}, \quad \ket{E_{k+1}(\mu(t')+\Delta\mu)}=\ket{v}.
    \end{equation*}
\end{theorem}
\Cref{theorem:eigenvector_swap} describes the scenario where a monotonically increasing schedule function $\mu(t)$ introduces the rank-one interaction Hamiltonian $|v\rangle \langle v|$. If $\ket{v}$ is almost parallel to one of the eigenvectors $\ket{d_k}$ of the initial Hamiltonian $D_{\mathrm{init}}$, then at a critical value $\mu(t')$, there will be an \mbox{$\epsilon$-avoided} level crossing where the instantaneous eigenvalues $E_k(\mu)$ and $E_{k+1}(\mu)$ rapidly exchange their eigenvectors. For sufficiently small $\epsilon$, this exchange is approximate; in the limit, it becomes exact.

A remarkable consequence of \Cref{theorem:eigenvector_swap} is that the eigenvector swap mechanism approximately continues at all other energy levels $E_{k+1}, E_{k+2}, \dots, E_{N}$; see \Cref{fig:edge_case_cascade}. As such, the condition \cref{eq:swap_condition} generates a cascade of eigenvalue swaps. We formalize this in the following corollary.

\begin{figure}[t]
    \centering
    \includegraphics[scale=0.15]{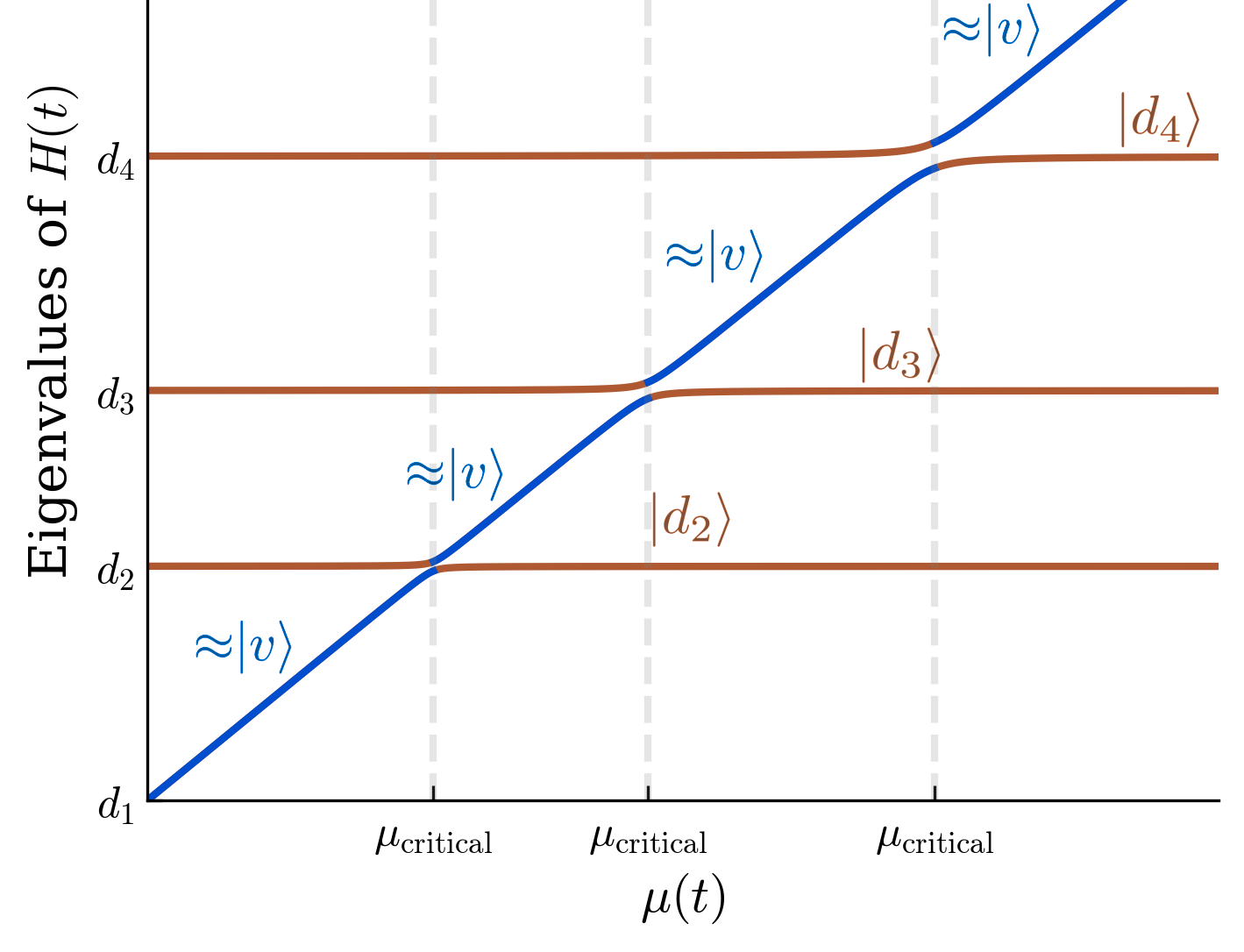}
    \caption{The cascade of eigenvector swaps described in \Cref{corollary:swap_cascade} for finite $\epsilon$. The blue curve denotes the instantaneous eigenvalue whose eigenvector is approximately $\ket{v}$. As the schedule $\mu(t)$ increases, this eigenvalue forms a sequence of \mbox{$\epsilon$-avoided} crossings. At each avoided level crossing, the corresponding instantaneous eigenvectors are approximately exchanged, producing a cascade of eigenvector swaps.}
\label{fig:edge_case_cascade}
\end{figure}

\begin{corollary}[Cascade of Eigenvector Swaps \cite{gabbassov2025adiabatic}]\label{corollary:swap_cascade}
    Under the assumptions of \Cref{theorem:eigenvector_swap}, the eigenvector exchange continues through all higher energy levels. In the limit $\epsilon\to0$, for each $m=k+1,\dots,N-1$, the instantaneous eigenvalues $E_m(\mu)$ and $E_{m+1}(\mu)$ exchange their eigenvectors at the \mbox{$\epsilon$-avoided} level crossing occurring at some $\mu(t_m)>0$. In the same limit, after the final avoided crossing, one has
    $$
    \ket{E_N(\mu)}=\ket{v}.
    $$
\end{corollary}
\Cref{corollary:swap_cascade} describes orderly eigenvector exchanges that occur between adjacent instantaneous eigenvalues at \mbox{$\epsilon$-avoided} level crossing. Namely, at the beginning of the evolution, the eigenvalue $E_k(\mu)$ acquires the eigenvector $\ket{v}$. Later at \mbox{$\epsilon$-avoided} level crossing, $E_k(\mu)$ and $E_{k+1}(\mu)$ exchange their eigenvectors, then $E_{k+1}(\mu)$ and $E_{k+2}(\mu)$ exchange their eigenvectors, and this continues until $\ket{v}$ becomes the eigenvector associated with highest energy $E_N(\mu)$.

The corresponding statement for a negative schedule is given in \Cref{corollary:eigenvector_swap_negative_schedule}.
    \begin{corollary}[Negative Schedule]\label{corollary:eigenvector_swap_negative_schedule}
    Under the assumptions of \Cref{theorem:eigenvector_swap}, but with $\mu(t)\in(-\infty,0]$ monotonically decreasing, the eigenvector exchange continues through all lower energy levels. In the limit $\epsilon \to 0$, for each $m=k,k-1,\dots,2$, the instantaneous eigenvalues $E_m(\mu)$ and $E_{m-1}(\mu)$ exchange their eigenvectors at the \mbox{$\epsilon$-avoided} level crossing occurring at some $\mu(t_m)<0$. In the same limit, after the final avoided crossing, one has
    $$
    \ket{E_1(\mu)}=\ket{v}.
    $$
\end{corollary}

In that case, the same mechanism occurs in the opposite direction. The eigenvector $\ket{v}$ is successively exchanged through the lower energy levels until, in the limit $\epsilon\to0$, it becomes the eigenvector associated with the ground energy level $E_1(\mu)$.

\noindent\textbf{Discussion.} \Cref{theorem:eigenvector_swap} and \Cref{corollary:swap_cascade} isolate the general and elementary spectral mechanics under the scheduled addition of interaction Hamiltonians. If the system remains in an instantaneous eigenstate, then partial or complete entanglement restructuring requires a corresponding partial or complete sequence of eigenvector swaps. The stronger these swaps are, the narrower the associated avoided level crossings must be. Consequently, implementing substantial entanglement changes requires the evolution to slow down considerably near the avoided level crossings. These insights suggest a strategy for faster evolution: if we can reduce unnecessary entanglement restructurings, we can accelerate the computation.

We will use \Cref{theorem:eigenvector_swap}, \Cref{corollary:swap_cascade}, and the insights developed in this section in the context of penalty-based and penalty-free quantum optimization algorithms.

\begin{figure*}[t]
    \centering
    \includegraphics[scale=.46]{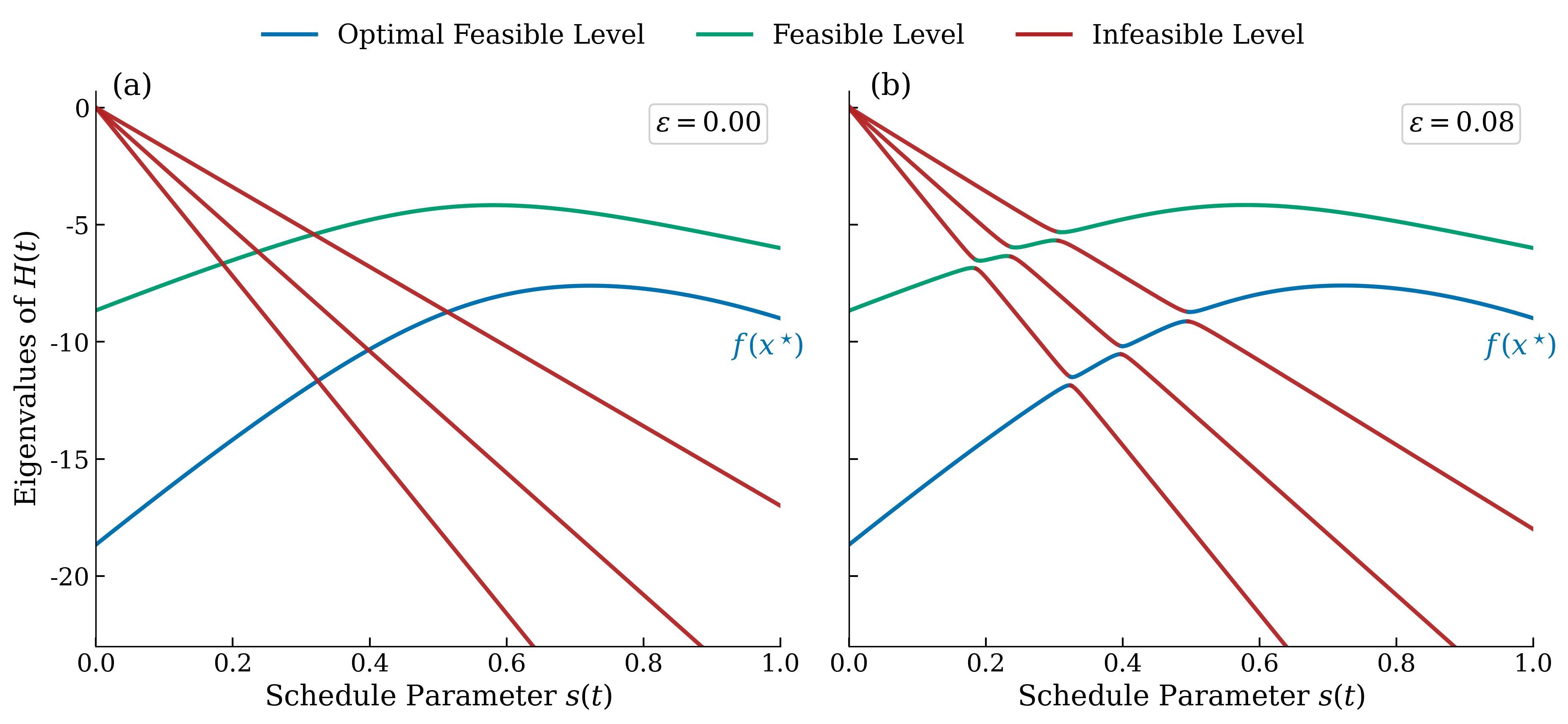}
    \caption{Relevant eigenvalue trajectories of the total Hamiltonian $H(t)$ for a 0--1 IP [see \Cref{ap:kp}]. The blue eigenvalue trajectory starts at the ground energy of the initial Hamiltonian at $s(t)=0$ and terminates at the optimal feasible value $f(x^\star)$ at $s(t)=1$. Hence, it is the desirable evolution trajectory for finding the optimal solution. The green curve is another feasible but suboptimal level. The red curves are lower-energy but infeasible levels. (a) For $\epsilon=0$, feasible and infeasible subspaces are invariant, so the blue trajectory crosses all infeasible levels exactly and reaches $f(x^\star)$ despite zero global gaps. (b) For $\epsilon=0.08$, weak coupling turns these crossings into narrow avoided crossings; in the fast-jump regime, the state jumps across them and again reaches $f(x^\star)$. Note that all infeasible levels are lower energy states of the problem Hamiltonian $H_{\mathrm{p}}$.}
\label{fig:spectral_gaps_knapsack}
\end{figure*}

\section{Quantum Optimization, Spectral Gaps and Entanglement}\label{sec:spectral_gaps}
This section applies the entanglement-dynamics framework to constrained quantum optimization. The central goal is to understand how problem constraints shape entanglement restructuring and how to control it to reduce computational bottlenecks that preclude speedups.

\subsection{Small Spectral Gaps That Enable Fast Evolution}
In this section, we use the eigenvector-swap mechanism of \Cref{theorem:eigenvector_swap} to show that narrow avoided level crossings need not constitute computational bottlenecks. Instead, narrowly avoided level crossings arising from weak coupling between the otherwise invariant subspaces $\mathcal H_{\mathcal F}$ and $\mathcal H_{\mathcal Q}$ can accelerate the computation. At such an avoided crossing, the approaching eigenvalues exchange their eigenvectors. If the goal is to remain on the same instantaneous energy level, the evolution must slow down considerably. Consequently, that level acquires a new eigenvector via the swap, and the system undergoes entanglement restructuring induced by the swapped-in eigenvector. If the goal is to avoid this restructuring, the evolution should instead speed up near the avoided crossing, causing the system to jump to the approaching energy level. Because the eigenvector swap and the dynamical jump occur at the same avoided crossing, the system jumps to the approaching energy level onto which its original eigenvector is swapped. \Cref{fig:eigenvector_swap_schematic} illustrates this behaviour: if the system starts at the level $E_k$ and jumps to the level $E_{k+1}$, it ends up with the original eigenvector. Thus, in this setting, small gaps can be used to avoid restructuring of entanglement and accelerate the computation.

It turns out that the mixer $H^{\epsilon}_{\mathrm{init}}$ in \cref{eq:relaxed_mixer} together with the problem Hamiltonian $H_{\mathrm{p}}$ encoding $f(x)$ in \cref{eq:canonical_ip} realize this \textit{fast-jump} regime. Specifically, we have the following setting:
\begin{enumerate}
    \item \label{item:1} The problem Hamiltonian $H_{\mathrm p}$ encodes the linear objective $f(x)$. As noted in \Cref{observation:infeasible_solution}, the lowest values of $f(x)$ are typically attained by infeasible bit strings. Hence, the ground eigenstate and many lower-energy eigenstates of $H_{\mathrm p}$ lie in $\mathcal H_{\mathcal Q}$, while the optimal feasible state $\ket{x^\star}$ appears at a higher energy level.
    
    \item \label{item:2} For $\epsilon=0$, we have \mbox{$H^{\epsilon}_{\mathrm{init}}=H^{0}_{\mathrm{init}} = H^{\mathrm{feas}}_{\mathrm{init}}$}. Therefore, the feasible and infeasible subspaces $\mathcal H_{\mathcal F}$ and $\mathcal H_{\mathcal Q}$ are invariant under the total Hamiltonian. Due to the invariance, eigenvalues whose eigenvectors are supported on different invariant subspaces cross exactly; see \Cref{fig:spectral_gaps_knapsack} (a).
    
    \item \label{item:3} For $\epsilon=0$, we assume that the restricted feasible Hamiltonian
    $$H_{FF}(t)=P_{\mathcal F}H(t)P_{\mathcal F} \nonumber$$
    has a continuous eigenvalue trajectory that starts from the chosen initial feasible eigenstate and terminates at the optimal feasible state $\ket{x^\star}$ at \mbox{$s(T)=1$}. This is the desired feasible trajectory shown in blue in \Cref{fig:spectral_gaps_knapsack} (a). This assumption holds for mixers that connect all feasible basis states through nonpositive off-diagonal matrix elements; concrete constructions are given in \cite{herman2023constrained}. Also, see examples in \cref{eq:case_study_mixer,eq:xy_mixer}.
    
    \item \label{item:4} For $\epsilon>0$, the terms $\epsilon H^{\mathrm{init}}_{FQ}$ and $\epsilon H^{\mathrm{init}}_{QF}$ weakly couple $\mathcal H_{\mathcal F}$ and $\mathcal H_{\mathcal Q}$. Hence, provided the corresponding coupling matrix elements are nonzero, the exact crossings associated with the $\epsilon=0$ generically open into $\epsilon$-avoided level crossings; see \Cref{fig:spectral_gaps_knapsack} (b). We assume that the two levels forming each $\epsilon$-avoided level crossing remain isolated from all other levels. This is natural when the $\epsilon$-avoided gaps vanish with $\epsilon$, while all other relevant gaps remain unchanged and asymptotically larger. Thus, for sufficiently small $\epsilon$, each such \mbox{$\epsilon$-avoided} crossing is locally a two-level avoided crossing.
\end{enumerate}

Therefore, for $\epsilon >0$, the desired eigenvalue trajectory [\Cref{fig:spectral_gaps_knapsack} (b), blue curve] is no longer continuous due to \mbox{$\epsilon$-avoided} level crossings. Therefore, if the system starts in the ground eigenstate of the initial Hamiltonian, it must jump upward at the \mbox{$\epsilon$-avoided} crossings, from one energy level to the next, until it reaches the level corresponding to the optimal feasible solution $\ket{x^\star}$. By the eigenvector-swap mechanism of \Cref{theorem:eigenvector_swap}, the energy levels exchange their eigenvectors at each such avoided crossing. Hence, when the system jumps to the approaching level, its state remains aligned with the original eigenstate rather than acquiring the swapped-in one. The fast jumps, therefore, allow the system to move upward in energy while avoiding entanglement restructuring

Assuming the setting above [\cref{item:1,item:2,item:3,item:4}], we formalize these dynamics in the following theorem.
\begin{theorem}[Fast Jumps Regime]\label{theorem:leakage_fast_jumps}
    Let $x^\star\in\mathcal F$ be an optimal feasible solution of 0--1 IP:
    \begin{equation*}
    f(x^\star)=\min_{y\in\mathcal F} f(y)
    \end{equation*}
    Let the total time-dependent Hamiltonian be
    \begin{equation}\label{eq:total_hamiltonian_in_theorem}
        H(t) = (1-s(t))H^{\epsilon}_{\mathrm{init}} + s(t) H_{\mathrm{p}}.
    \end{equation}
    For $\epsilon$ sufficiently small, if the system starts in the ground state of $H^\epsilon_{\mathrm{init}}$ and evolves sufficiently fast, then it jumps upward at the \mbox{$\epsilon$-avoided} level crossings, from one instantaneous energy level to the next, while preserving its eigenvector and entanglement structure. Consequently, it reaches an energy level $E_k(t)$ such that
    \begin{equation}
        E_k(T)=f(x^\star),
    \end{equation}
    with eigenvector $\ket{x^\star}$.
\end{theorem}
\begin{proof}
    Let us write down the IP:
    \begin{equation}\label{eq:canonical_ip_2}
        \begin{aligned}
            \mathrm{min } \quad & c^{\mathsf T} x \\
            \mathrm{subject\ to } \quad & Ax \geq b\\
            & \ \ x \in \{0,1\}^n
        \end{aligned}
    \end{equation}
    Assume, without loss of generality, that $f(x)<0$ for every $x\in\{0,1\}^n$ and $x^{\star}$ is an optimal solution. The negative $f(x)$ can be achieved by a scalar offset which does not affect the dynamics. Assume \mbox{$z \notin \mathcal{F}$} be an infeasible solution such $f(z) < f(x^{\star})$. 

    First, consider the case $\epsilon=0$. Then 
    \begin{equation}
        H^{\epsilon}_{\mathrm{init}} = H^{0}_{\mathrm{init}} = H^{\mathrm{feas}}_{\mathrm{init}},
    \end{equation}
    where $H^{\mathrm{feas}}_{\mathrm{init}}$ is feasibility-preserving mixer
    defined in \cref{eq:feasible_subspace_mixer,eq:mixer_kernel}. By construction, $\ket{z}$ is an excited eigenstate of the mixer,
    \begin{equation*}
        H^{0}_{\mathrm{init}}\ket{z} = 0\ket{z}.
    \end{equation*}
    It follows that $\ket{z}$ is an instantaneous eigenvector of the total Hamiltonian with the eigenvalue $E_z(t) = s(t)f(z)$,
    \begin{align}\label{eq:z_is_eigenvec}
        H(t)\ket{z} = s(t)H_{\mathrm{p}}\ket{z} = s(t)f(z)\ket{z}.
    \end{align}
    Let $E_{\mathcal F}(t)$ be the eigenvalue of $H(t)$ that starts at the ground energy of $H_{\mathrm{init}}^0$ and terminates at $E_{\mathcal F}(T)=f(x^\star)$. Since $f(z)< f(x^\star)$, $E_z(t)$ starts above $E_{\mathcal F}(t)$ and terminates below it. That is,
    \begin{equation*}
        E_{\mathcal F}(0)<E_z(0)=0, 
    \end{equation*}
    and
    \begin{equation*}
        f(z) = E_z(T) < E_{\mathcal F}(T) = f(x^\star).
    \end{equation*}
    Therefore, $E_z(t)$ and $E_{\mathcal F}(t)$ have an exact level crossing. Crossing follows from subspace invariance,
    \begin{equation*}
        H(t) \mathcal{H}_{\mathcal{Q}} \subseteq \mathcal{H}_{\mathcal{Q}}.
    \end{equation*}
    For $\epsilon>0$, the off-diagonal terms $\epsilon H^{\mathrm{init}}_{FQ}$ and $\epsilon H^{\mathrm{init}}_{QF}$ break this invariance and weakly couple the two subspaces. Hence, the corresponding exact crossings generically open into $\epsilon$-avoided level crossings.
    
    Our goal is to relate the proof to \Cref{theorem:eigenvector_swap} locally near $\epsilon$-avoided level crossing. Then all statements of entanglement restructuring follow naturally. To this end, for $s(t)<1$, rescale the Hamiltonian $H(t)$ by the positive factor $(1-s(t))^{-1}$ and rewrite $H_{\mathrm{p}}$ in its eigendecomposition. This yields
    \begin{align*}
        \widehat H(t) &= H_{\mathrm{init}}^\epsilon + \sum_{x \in \{0,1\}^n} \mu_x(t)\ket{x}\bra{x}, \\
        \mu_x(t) &= \mu(t) f(x), \quad  \mu(t) = \frac{s(t)}{1-s(t)}.
    \end{align*}
    The rescaling preserves the instantaneous eigenvectors and their energy ordering. Moreover, since $f(x)<0$, every $\mu_x(t)$ is negative and monotonically decreasing.
    
    Fix a relevant infeasible solution $z$ and set
    \begin{equation*}
        \ket v=\ket z.
    \end{equation*}
    Let $t'$ be the time at which, for $\epsilon=0$, the eigenvalue associated with $\ket z$ crosses the feasible eigenvalue being followed. Near this crossing, we write
    \begin{equation}\label{eq:local_rank_one_fast_jump}
        \widehat H(t) = D_{\mathrm{init}}(t') + \mu_z(t)\ket v \bra v + R_z(t),
    \end{equation}
    where
    \begin{equation}\label{eq:local_D_init_fast_jump}
    D_{\mathrm{init}}(t')
    = H_{\mathrm{init}}^\epsilon +\sum_{x\neq z}\mu_x(t')|x \rangle \langle x|,
    \end{equation}
    and
    \begin{equation}
        \begin{aligned}
            R_z(t) &= \sum_{x\neq z} \bigl(\mu_x(t)-\mu_x(t')\bigr)|x \rangle \langle x| \\
            &= \sum_{x\neq z} \bigl(\mu(t)-\mu(t')\bigr) f(x) |x \rangle \langle x| \\
            &=  \sum_{x\neq z} \Delta \mu(t) f(x)|x \rangle \langle x| \\
            & = \Delta \mu(t) \sum_{x\neq z} f(x)|x \rangle \langle x| = O(\Delta \mu(t)).
        \end{aligned}
    \end{equation}
    The first two terms in \cref{eq:local_rank_one_fast_jump} have precisely the rank-one form required by \Cref{theorem:eigenvector_swap}. For $\Delta \mu(t) = \mu(t)-\mu(t')$ small, $R_z(t)$ is also small. Thus, in a sufficiently small neighbourhood of the avoided crossing, the leading local dynamics are governed by the rank-one Hamiltonian
    $$D_{\mathrm{init}}(t')+\mu_z(t)| v\rangle \langle v|,$$
    while $R_z(t)$ gives perturbative corrections controlled by the size of the neighbourhood.
    
    Let \mbox{$d_1<d_2<\cdots<d_N$} be the eigenvalues of $D_{\mathrm{init}}(t')$, with normalized eigenvectors $\ket{d_i}$. For sufficiently small $\epsilon>0$, and assuming that the relevant weak-coupling is nonzero, there is an eigenvector $\ket{d_j}$ such that
    \begin{equation*}
        |\langle v|d_j\rangle|^2=1-O \left ( \epsilon^2 \right ),\ \ |\langle v | d_i \rangle|^2 = O \left ( \epsilon^2 \right ), \ \text{for every }i\neq j.
    \end{equation*}
    Eigenvector $\ket{d_j}$ exists by the construction of the mixer $H^{\epsilon}_{\mathrm{init}}$. That is, one can show that
    \begin{equation}
        \ket{d_j} = \alpha \ket{v} + \beta \ket{\zeta}, \text{ with } \ket{\zeta} \perp \ket{v},
    \end{equation}
    and \mbox{$|\alpha|^2 = 1 - O(\epsilon^2)$}. Note that for \mbox{$\epsilon = 0$}, we have \mbox{$\ket{d_j} = \ket{v}$} and hence \mbox{$|\alpha|^2 = 1$}. This is because, by construction, \mbox{$H^{0}_{\mathrm{init}}\ket{v} = 0 \ket{v}$}. Thus, up to perturbative corrections due to $R_z(t)$, the local Hamiltonian near the crossing has the rank-one form required by \Cref{theorem:eigenvector_swap}, with \mbox{$\ket v=\ket z$} and schedule $\mu_z(t)$. Because \mbox{$\mu_z(t)<0$} is monotonically decreasing, \Cref{corollary:eigenvector_swap_negative_schedule} applies locally. Therefore, near this isolated \mbox{$\epsilon$-avoided} level crossing, the eigenvector $\ket v$ is exchanged with the approaching lower energy level. The exchange occurs over the short interval of $2 \Delta \mu$; hence, $R_z(t)$ is negligible around the avoided level crossing.
    
    Consider one such \mbox{$\epsilon$-avoided} level crossing between the adjacent instantaneous eigenvalues $E_{m-1}$ and $E_m$, occurring near $\mu_z(t')<0$. Since $\mu_z(t)$ is decreasing, $\mu_z(t')+\Delta\mu$ denotes a value immediately before the \mbox{$\epsilon$-avoided} level crossing, while $\mu_z(t')-\Delta\mu$ denotes a value immediately after it. In the limit $\epsilon \to 0$, the negative-schedule eigenvector swap gives (up to phases)
    \begin{equation}\label{eq:negative_schedule_swap_fast_jump}
        \ket{E_{m-1}\bigl(\mu_z(t')+\Delta\mu\bigr)} = \ket{E_m\bigl(\mu_z(t')-\Delta\mu\bigr)}.
    \end{equation}
    In the fast-jump regime, the system's state $\ket{\psi(\mu)}$ jumps from the level $m-1$ to $m$, but changes negligibly over the narrow interval of the avoided level crossing. That is, for $\epsilon \to 0$, we have:
    \begin{align*}
        &\text{Before: } \ket{\psi\bigl(\mu_z(t')+\Delta\mu\bigr)} = \ket{E_{m-1}\bigl(\mu_z(t')+\Delta\mu\bigr)},\\
        &\text{After: } \ket{\psi\bigl(\mu_z(t')-\Delta\mu\bigr)} = \ket{E_m\bigl(\mu_z(t')-\Delta\mu\bigr)}.
    \end{align*}
    The change is negligible because of \cref{eq:negative_schedule_swap_fast_jump}. Hence, after the avoided level crossing, at $\epsilon \to 0$, the later system state is
    \begin{equation*}
        \ket{\psi\bigl(\mu_z(t')-\Delta\mu\bigr)} = \ket{E_{m-1}\bigl(\mu_z(t')+\Delta\mu\bigr)}.
    \end{equation*}
    Because the jump and the eigenvector swap occur simultaneously, the system retains its original eigenvector and the associated entanglement structure.
    
    Repeating this local argument for every isolated infeasible level terminating below $f(x^\star)$ keeps the system's state on the feasible eigenvector trajectory while increasing only its energy-level index. At $t=T$, this trajectory terminates at the lowest feasible eigenvector of $H_{\mathrm p}$, namely $\ket{x^\star}$. Consequently, for $\epsilon \to 0$, for some $k$, we have
    \begin{equation*}
        E_k(T)=f(x^\star), \qquad \ket{E_k(T)}=\ket{x^\star}.
    \end{equation*}
    We also note that the perturbative error due to $R_z(t) = O(\Delta \mu)$ vanishes with $\epsilon$ because the interval $\Delta \mu$ can be made progressively smaller with $\epsilon$.
\end{proof}

\Cref{theorem:leakage_fast_jumps} states that the desired energy level can be reached rapidly without significant restructuring of entanglement. We now estimate how short the total evolution time $T$ must be for the system to jump the $\epsilon$-avoided crossings.
\begin{corollary}[Sufficient Fast-Jump Condition]\label{corollary:fast_evolution}
    Assume a linear schedule, $s(t)=t/T$, where $T$ is the total duration. Let $g(\epsilon)$ denote the largest actual gap among $\epsilon$-avoided level crossings. Suppose that this gap is bounded above at order $\epsilon$. Then a sufficient condition for diabatic passage through the $\epsilon$-avoided level crossings is
    \begin{equation}
        T\ll \epsilon^{-2}.
    \end{equation}
\end{corollary}
\begin{proof}
    Adiabatic passage, corresponding to remaining on the same instantaneous energy level, requires
    \begin{equation}
    T \gg g(\epsilon)^{-2}.
    \end{equation}
    The fast-jump regime, corresponding to diabatic passage through the weak avoided crossing, is the opposite limit,
    \begin{equation}
    T \ll g(\epsilon)^{-2}.
    \end{equation}
    Since the gap is bounded above at order $\epsilon$, there exists a constant $C$ such that
    $$g(\epsilon) < C\epsilon.$$
    It follows that,
    $$C^{-2} \epsilon^{-2} < g(\epsilon)^{-2}.$$
    Therefore, the stronger sufficient condition for a diabatic passage through the $\epsilon$-avoided level crossings is
    \begin{equation}
        T\ll \epsilon^{-2}.
    \end{equation}
    If a particular avoided-crossing gap is smaller than order $\epsilon$, the admissible fast-jump time can be longer; nevertheless, \mbox{$T\ll\epsilon^{-2}$} still enforces diabatic passage.
\end{proof}

Not every avoided crossing should be jumped. Once the state has reached the energy level $E_k(t)$ that terminates at the feasible optimum, it must remain on this level. Let $\Delta$ denote the minimum gap the system must not jump. By the adiabatic theorem, remaining on the same instantaneous energy level requires
\begin{equation}
\Delta^{-2} \ll T.
\end{equation}
On the other hand, the weak-coupling avoided crossings require
\begin{equation}
T \ll \epsilon^{-2}.
\end{equation}
If $\epsilon \ll \Delta$, these two requirements are compatible, and one may choose
\begin{equation}
\Delta^{-2} \ll T \ll \epsilon^{-2}.
\end{equation}
In this regime, the evolution is slow enough to remain on the same instantaneous energy level across avoided crossings with gaps at least $\Delta$, but fast enough to jump the \mbox{$\epsilon$-avoided} crossings.

\subsection{Zero Spectral Gap Is Not a Bottleneck}\label{sec:zero_gap_not_bottleneck}
In the previous section, we showed that the narrow avoided-level crossings need not slow the evolution and can facilitate faster dynamics. In this section, we show that an exactly closed gap, corresponding to a true level crossing, need not constitute a computational bottleneck. This insight again suggests that the spectral gap alone is an incomplete indicator of computational slowdown.

Let us consider the penalty-free mixer $H^{\text{feas}}_{\text{init}}$ in \cref{eq:feasible_subspace_mixer} and its problem Hamiltonian $H_{\text{p}}$. This mixer only allows transitions within the feasible subspace $\mathcal{H}_{\mathcal{F}}$, which is spanned by the feasible computational basis states corresponding to $\mathcal{F}$ in \cref{eq:feasible_set_binary}. The problem Hamiltonian $H_{\text{p}}$, which encodes a linear objective function $f(x)$ in \cref{eq:canonical_ip}, is diagonal in the computational basis. Therefore, neither Hamiltonian supports transitions outside the feasible subspace. In other words, the feasible set induces a decomposition of the Hilbert space as $\mathcal{H}=\mathcal{H}_{\mathcal{F}} \, \oplus \, \mathcal{H}_{\mathcal{Q}}$ with $\mathcal{H}_{\mathcal{F}}\perp\mathcal{H}_{\mathcal{Q}}$, and the total Hamiltonian is block diagonal with respect to this decomposition,
\begin{equation*}
    P_{\mathcal{F}}H(t)P_{\mathcal{Q}} = P_{\mathcal{Q}}H(t)P_{\mathcal{F}} = 0,
\end{equation*}
and hence,
\begin{equation*}
    H(t) =
     \begin{pmatrix}
        P_{\mathcal{F}}H(t)P_{\mathcal{F}} && 0\\
        0 && P_{\mathcal{Q}}H(t)P_{\mathcal{Q}}
    \end{pmatrix}.
\end{equation*}
Thus, the two subspaces are invariant under the dynamics,
\begin{equation}\label{eq:subspace_invariance}
    H(t)\mathcal{H}_{\mathcal{F}}\subseteq\mathcal{H}_{\mathcal{F}},\qquad H(t)\mathcal{H}_{\mathcal{Q}}\subseteq\mathcal{H}_{\mathcal{Q}}.
\end{equation}
This leads to an interesting proposition.
\begin{proposition}[Level Crossings]\label{prop:level_crossing}
    Suppose the total Hamiltonian
    \begin{equation}
        H(t) = (1-s(t))H^{\mathrm{feas}}_{\mathrm{init}} + s(t)H_{\mathrm{p}}
    \end{equation}
    encodes the 0--1 IP in \cref{eq:canonical_ip}. If the IP has at least one infeasible solution $z \notin \mathcal{F}$ such that $f(z) < f(x)$ for all $x \in \{0,1\}^n$ and $x\neq z$, then the total Hamiltonian $H(t)$ has level crossings and zero minimum spectral gap.
\end{proposition} 
\begin{proof}
    This is a special case of \Cref{theorem:leakage_fast_jumps} with $\epsilon = 0$, which is covered in the proof of \Cref{theorem:leakage_fast_jumps}.
    Let $z \notin \mathcal{F}$ be an infeasible solution such $f(z) < f(x)$ for all $x \in \{0,1\}^n$ and $x \neq z$. It follows that $\ket{z}$ is the ground state of the problem Hamiltonian $H_{\mathrm{p}}$,
    \begin{equation}
        H_{\mathrm{p}}\ket{z} = f(z) \ket{z}.
    \end{equation}
    Furthermore, by the construction of the mixer $H^{\mathrm{feas}}_{\mathrm{init}}$ in \cref{eq:feasible_subspace_mixer,eq:mixer_kernel}, $\ket{z}$ is an excited eigenstate of the mixer,
    $$H^{\mathrm{feas}}_{\mathrm{init}}\ket{z} = 0\ket{z}.$$
    It follows that $\ket{z}$ is an instantaneous eigenvector of the total Hamiltonian with the eigenvalue $E_z(t) = s(t)f(z)$,
    \begin{align*}
        H(t)\ket{z} = s(t)H_{\mathrm{p}}\ket{z} = s(t)f(z)\ket{z}.
    \end{align*}
    This is consistent with the subspace invariance given by \cref{eq:subspace_invariance}, that is, $H(t) \mathcal{H}_{\mathcal{Q}} \subseteq \mathcal{H}_{\mathcal{Q}}$.
    Finally, we note that $\ket{z}$ is an excited eigenstate of $H^{\mathrm{feas}}_{\mathrm{init}}$ and it is also the ground eigenstate of $H_{\mathrm{p}}$.
    Therefore, it must be true that the instantaneous eigenvalue $E_z(t) = s(t)f(z)$ crosses at least one excited energy level as it descends to the ground energy $s(T)f(z) = f(z)$. See \Cref{fig:spectral_gaps_knapsack} (a) for the demonstration.
\end{proof}

\noindent \textbf{Discussion.} The gap closing in Proposition~\ref{prop:level_crossing} is not a computational bottleneck. It is a crossing between eigenvalues whose eigenvectors belong to orthogonal and invariant subspaces. Thus, a state initialized in $\mathcal H_{\mathcal F}$ cannot transition into $\mathcal H_{\mathcal Q}$ under the evolution. Therefore, the relevant spectral gap is the gap of the restricted Hamiltonian
\begin{equation*}
    H_{FF}(t)=P_{\mathcal F}H(t)P_{\mathcal F},
\end{equation*}
rather than the global gap of $H(t)$. Hence, a global gap closing caused by a crossing between an eigenvalue associated with $\mathcal H_{\mathcal Q}$ and an eigenvalue associated with $\mathcal H_{\mathcal F}$ does not by itself imply a computational bottleneck.

\section{Penalty-Free and Penalty-Based Methods as Spectral Duals}\label{sec:spectral_duality}
In the previous sections, we encountered a rather paradoxical dynamics: the system starts in the ground state of the initial Hamiltonian but terminates in an excited eigenstate of the problem Hamiltonian, corresponding to the optimal solution of a 0--1 IP. As we have seen, in penalty-free methods, the state ascends through the energy levels until it reaches the level associated with the optimal objective value; e.g., see \Cref{fig:spectral_gaps_knapsack}. This is the inherent behaviour of penalty-free methods because many infeasible solutions $z$ tend to have lower objective function values than the feasible ones:
$$f(z) \leq f(x), \ \text{ for all } \ x \in \mathcal{F}$$
Consequently, the ground eigenstate of the problem Hamiltonian corresponds to the infeasible $\ket{z}$ and an excited eigenstate corresponds to the optimal solution $\ket{x^{\star}}$.

\begin{figure*}[t]
    \centering
    \includegraphics[scale=.5]{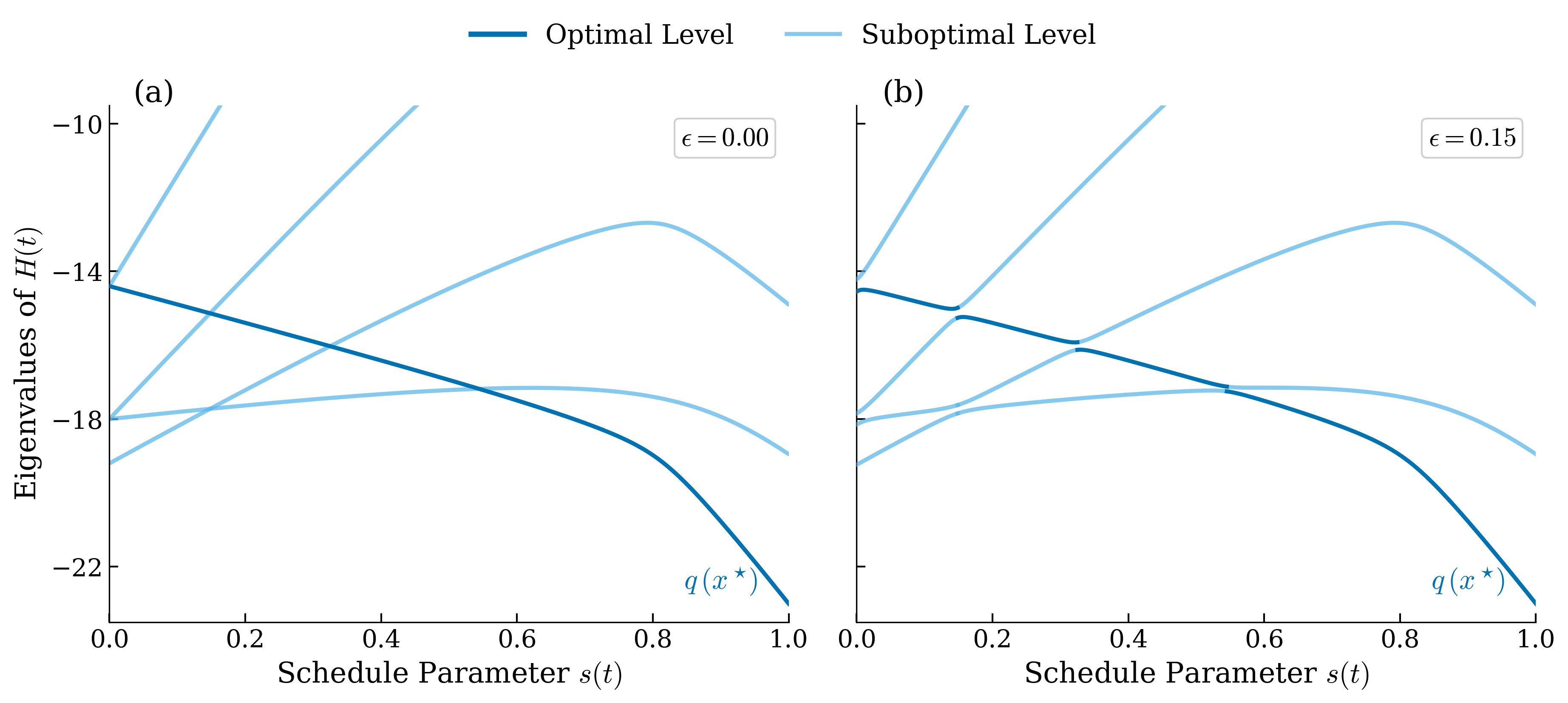}
    \caption{Relevant eigenvalue trajectories of the total Hamiltonian $H(t)$ for a penalty-based encoding of a 0--1 IP [see, \Cref{ap:x3c}]. The blue eigenvalue trajectory starts from an excited energy level of the initial Hamiltonian at $s(t)=0$ and terminates at the ground energy $q(x^\star)$ of the problem Hamiltonian at $s(t)=1$. Hence, it is the desirable evolution trajectory for finding the optimal solution. The light-blue curves are suboptimal levels. (a) For $\epsilon=0$, the invariant subspaces are decoupled, so the blue trajectory crosses the suboptimal levels exactly and reaches $q(x^\star)$ despite zero global gaps. (b) For $\epsilon=0.15$, weak coupling turns these exact crossings into narrow avoided crossings; in the fast-jump regime, the state jumps across them and again reaches $q(x^\star)$. This is the spectral dual of \Cref{fig:spectral_gaps_knapsack}; i.e., the system's energy descends to the ground energy of $H_{\mathrm{p}}$, rather than ascending as in the previous case.}
\label{fig:spectral_gaps_exact_cover}
\end{figure*}

Here, we will show that penalty-based methods exhibit the opposite spectral behaviour. In a general sense, they are spectral duals of penalty-free methods: instead of ascending toward an excited eigenstate of the problem Hamiltonian, the state descends toward its ground energy level. In the traditional penalty-based setting, this descent is absent because the system is initialized in the ground state of the initial Hamiltonian and remains at the ground energy level throughout the evolution. More generally, however, one need not initialize in the ground state. The system may start in a suitably chosen excited eigenstate and then descend through the energy levels until it reaches the ground energy level of the problem Hamiltonian. The descent through energy levels is an inherent behaviour of the penalty-based methods. Because for all $x$, relevant slack variables $w$, and large penalty $\lambda$, we have
$$q(x^{\star}, w^{\star}) \leq q(x, w)$$
Hence, the ground eigenstate of the problem Hamiltonian is always an optimal solution. Thus, to reach it, we must either start at the ground energy level or descend to it.

Although starting from an excited state of the mixer may appear counterintuitive, it is both well motivated and more general. It is well-motivated because problem-specific information or constraints may identify an invariant subspace that contains, or is expected to contain, the optimal solution. Initializing in a mixer eigenstate supported on this invariant subspace may significantly reduce the search space.  It is more general because the initialization need not be tied to the ground eigenstate of the initial Hamiltonian, and the evolution need not be purely adiabatic or diabatic.

\subsection{Spectral Duality Example}
Let us consider an example in which the system starts in an excited eigenstate of a mixer and evolves toward the ground eigenstate of the problem Hamiltonian $H_{\text{p}}$, see \Cref{fig:spectral_gaps_exact_cover}. Suppose that an $n$-variable objective function $q(x)$ is encoded into a diagonal problem Hamiltonian $H_{\text{p}}$ acting on the $2^n$-dimensional Hilbert space $\mathcal{H}$. Assume that prior information indicates that the optimal solution $x^\star$ has Hamming weight $k$. Or alternatively assume that there is an implicit constraint $\sum_{i=1}^n x_i = k$. This suggests choosing a mixer that preserves Hamming weight, so that the evolution is confined to the part of the Hilbert space containing $\ket{x^\star}$. We therefore choose the XY mixer \cite{rieffel2020xy,hadfield2019quantum},
\begin{equation}\label{eq:xy_mixer}
    H_{\text{init}}^{\text{XY}} = -\sum_{(i,j)\in S} \left( X_iX_j + Y_iY_j \right),
\end{equation}
where the set $S$ specifies the pairs of interacting qubits. This mixer preserves Hamming weight. Since $H_{\text{p}}$ is diagonal in the computational basis, the total Hamiltonian preserves the decomposition
\begin{equation*}
    \mathcal{H}=\mathcal{H}_{w_0} \oplus \mathcal{H}_{w_1} \oplus \cdots \oplus \mathcal{H}_{w_n},
\end{equation*}
where $\mathcal{H}_{w_j}$ is spanned by computational basis states of Hamming weight $j$. Hence, if $x^\star$ has Hamming weight $k$, the search may be restricted to the invariant subspace $\mathcal{H}_{w_k}$. For fixed $k$ and arbitrarily large $n$, this subspace has dimension
\begin{equation*}
    \dim \mathcal{H}_{w_k} = \binom{n}{k} = O(n^k),
\end{equation*}
Thus, for a fixed $k$, the choice of mixer reduces the effective search space from exponential to polynomial size in $n$.

\begin{figure*}[t]
    \centering
    \includegraphics[scale=.25]{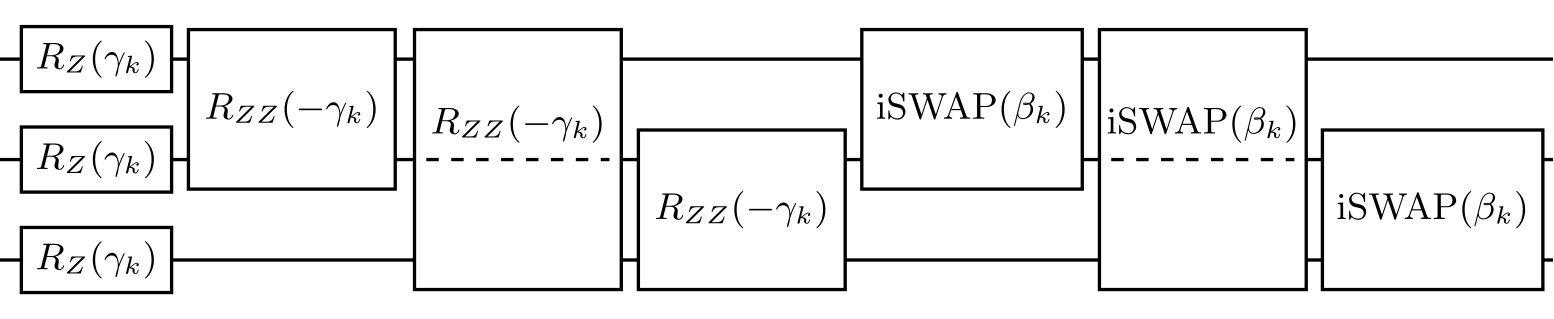}
    \caption{A $k$th layer of a 3-qubit quantum circuit that approximates the evolution generated by the total Hamiltonian \mbox{$H(t) = (1-s(t)) H^{\mathrm{XY}}_{\mathrm{init}} + s(t) H_{\mathrm{p}}$}, where $H_{\mathrm{p}}$ encodes a penalty-based objective function $q(x)$. $R_Z(\theta)$ and $R_{ZZ}(\theta)$ are one- and two-qubit phase gates parametrized by angle $\theta$. These gates are due to $Z_i$ and $Z_i Z_j$ interactions in $H_{\mathrm{p}}$. The $\mathrm{iSWAP}(\theta)$ is a parametrized two-qubit gate. This gate is due to the $X_iX_j+Y_iY_j$ interactions in $H^{\mathrm{XY}}_{\mathrm{init}}$.}
\label{fig:quantum_circuit}
\end{figure*}

It remains to choose the initial state. Since $\ket{x^\star}$ is in $\mathcal{H}_{w_k}$, we initialize the system in an eigenstate $\ket{\psi_{\text{init}}}$ of $H_{\text{init}}^{\text{XY}}$ supported on the same invariant subspace. In particular, we choose $\ket{\psi_{\text{init}}}$ to be the ground eigenstate of the restricted mixer:
\begin{equation*}
    H_{\text{init}}^{\text{XY}}\big|_{\mathcal{H}_{w_k}}
\end{equation*}
This initial eigenstate is typically an excited eigenstate of the full mixer. Nevertheless, within the invariant subspace $\mathcal{H}_{w_k}$, the evolution starts from the ground eigenstate of the restricted mixer and ends at the ground eigenstate $\ket{x^\star}$ of the problem Hamiltonian.

If $S$ is the edge set of the complete graph, then the Dicke state
\begin{equation*}
    \ket{D_k^n}=\frac{1}{\sqrt{\binom{n}{k}}} \sum_{x:\, |x|=k} \ket{x}
\end{equation*}
is the ground eigenstate of the restricted mixer on $\mathcal{H}_{w_k}$ which also admits efficient preparation \cite{bartschi2019deterministic,bartschi2022short}.

One of the many quantum algorithms that realizes this evolution can be obtained by discretizing the Schr\"odinger evolution generated by $H(t)$ and applying the first-order Lie-Trotter-Suzuki approximation \cite{suzuki1990fractal,ostmeyer2023optimised}. The result is an approximating quantum circuit depicted in \Cref{fig:quantum_circuit}.

\textbf{Discussion.}
In \Cref{prop:level_crossing}, we showed that the penalty-free approach produces exact level crossings in the spectrum of the total Hamiltonian $H(t)$. These crossings arise because the initial ground energy level ascends to the final excited energy level corresponding to $f(x^\star)$. The example above exhibits the spectral dual behaviour: an initial excited energy level descends to the final ground energy level corresponding to $q(x^\star,w^\star)$. During this descent, the instantaneous energy level must either cross intermediate levels belonging to invariant orthogonal subspaces or undergo avoided crossings, as shown in \Cref{fig:spectral_gaps_exact_cover} (a). Thus, \Cref{prop:level_crossing} extends directly to the penalty-based setting, with the direction of motion through the spectrum reversed.

The same duality also extends \Cref{theorem:leakage_fast_jumps} to the penalty-based method when an $\epsilon$-weak coupling is introduced between otherwise invariant subspaces, as shown in \Cref{fig:spectral_gaps_exact_cover} (b). In this case, the exact crossings become $\epsilon$-controlled avoided level crossings . Fast evolution can then jump over these avoided crossings, allowing the state to remain close to its original configuration rather than slowly following the instantaneous eigenvector exchanged at each crossing. Consequently, in both penalty-free and penalty-based methods, small couplings can be used to reduce entanglement restructuring while still reaching the energy level corresponding to the optimal solution.

\section{State Transitions}\label{sec:state_transitions}
We now show that the penalty-free and penalty-based methods induce substantially different state transitions, resulting in distinct wave propagation dynamics. Although both methods enforce feasibility, they weight allowable transitions between feasible states according to different criteria. Identifying these criteria provides a systematic basis for the design and analysis of quantum optimization algorithms.

We begin by classifying the relevant transitions. The mixer Hamiltonian
$H_{\text{init}}$ defines a transition graph whose vertices are computational
basis states and whose edges are the transitions generated by the mixer. Let
$\mathcal{F}$ be the feasible solution set defined in
\cref{eq:feasible_set_binary}, and let $x,y \in \mathcal{F}$.

\begin{definition}[First-order state transition]\normalfont
If $x$ and $y$ are adjacent in the transition graph induced by
$H_{\text{init}}$, i.e., if $(x,y)$ is an edge, then a
\textit{first-order state transition} is a direct transition  from $\ket{x}$ to $\ket{y}$:
\begin{equation*}
    \text{first-order state transition:} \ \ \ket{x} \rightarrow \ket{y}
\end{equation*}
\hfill $\diamond$
\end{definition}

\begin{definition}[Second-order state transition]\normalfont
If $x$ and $y$ are not adjacent but are connected by an infeasible intermediate
state $z\notin\mathcal{F}$, i.e., if $(x,z)$ and $(z,y)$ are edges, then a
\textit{second-order state transition} is an indirect transition mediated by
$\ket{z}$:
\begin{equation*}
    \text{second-order state transition:} \ \ \ket{x} \rightarrow \ket{z} \rightarrow \ket{y}
\end{equation*}
\hfill $\diamond$
\end{definition}

We show that penalty-based formulations induce second-order transitions within the feasible subspace. Specifically, two feasible states $\ket{x}$ and $\ket{y}$ that are not directly connected by the mixer become coupled through a virtual transition via an infeasible intermediate state $\ket{z}$. The strength of the second-order transition is scaled by the inverse of the energy penalty assigned to $\ket{z}$. Namely, the weighting factor is
\begin{equation}\label{eq:lambda_weight_factor}
    W_z^{\lambda} = \frac{1}{\lambda g(z,w)},
\end{equation}
where $g(z,w)$ is a penalty function in \cref{eq:penalty_function} and \mbox{$\lambda \gg 1$} is a large penalty multiplier.  If the intermediate infeasible state $\ket{z}$ strongly violates the constraint, then the second-order transitions through it are suppressed. On the other hand, if $\ket{z}$ mildly violates the constraints (almost on the boundary of being feasible), then transitions through it are enhanced.

Penalty-free formulations with the relaxed mixer $H_{\text{init}}^{\epsilon}$ in \cref{eq:relaxed_mixer} can also induce second-order state transitions through infeasible intermediate states. However, these transitions are not weighted by a constraint penalty. Instead, for an infeasible mediator $\ket{z}$, the second-order transition is weighted by 
\begin{equation}\label{eq:epsilon_weight_factor}
    W_z^{\epsilon}(t) = \frac{\epsilon^2}{f(z) - E(t)/s(t)},
\end{equation}
where $E(t)$ is the instantaneous energy of the total Hamiltonian $H(t)$ in \cref{eq:total_hamiltonian}, $s(t) > 0 $ is the schedule function, and positive $\epsilon \ll 1$. We can immediately see that second-order transitions are controlled by the objective function value $f(z)$ of the mediator $\ket{z}$ relative to the scaled instantaneous energy $E(t)$, rather than by an explicit penalty term. If the difference $f(z)-E(t)/s(t)$ is large, then the corresponding second-order transition is suppressed. If this difference is small, the transition is enhanced. Thus, the dominant second-order state transitions are those with $f(z) \approx E(t)/s(t)$. We note that for $t = T$ and $E(T) = f(z)$, we have $s(T)=1$ and $f(z) - E(T) = 0$. In this case, the system has left the feasible subspace and terminated at the infeasible eigenstate $\ket{z}$. In such a case of failure, the result \cref{eq:epsilon_weight_factor} does not apply.

It is enlightening to further analyze the second-order state transitions of the penalty-free methods. In optimization terms, low-cost infeasible states (including highly infeasible ones) can mediate transitions between feasible configurations. This can enhance exploration by creating virtual shortcuts through highly infeasible regions, but it can also allow strongly infeasible low-cost states to influence the dynamics within the feasible subspace. To see this, let us suppose that transition mediator $\ket{z}$ strongly violates the constraint, but it has a favourable objective value $f(z)$ which is close to $E(t)/s(t)$. Then, second-order transitions to feasible states $\ket{y}$ mediated by $\ket{z}$ are enhanced. This is the opposite of the penalty-based methods, where second-order state transitions are suppressed when the mediator $\ket{z}$ is highly infeasible. Now, suppose $\ket{z}$ is mildly infeasible, but its objective function value $f(z)$ is severely suboptimal. Then the difference $f(z) - E(t)/s(t)$ is large. Hence, the transitions mediated by such $\ket{z}$ are suppressed. Again, this is the opposite of the penalty-based method, where if $\ket{z}$ is only mildly infeasible, the penalty is low, and hence second-order state transitions are enhanced.

\begin{proof}
We show how the weight factors in \cref{eq:lambda_weight_factor,eq:epsilon_weight_factor} are derived. The main idea is to derive the effective Hamiltonian governing first- and second-order state transitions within the feasible subspace. Then, show that all second-order state transitions are weighted by factors defined in \cref{eq:lambda_weight_factor} and \cref{eq:epsilon_weight_factor}. We commence with the eigenvalue equation
\begin{equation}
    H(t) \ket{\psi(t)} = E(t) \ket{\psi(t)},
\end{equation}
where $H(t)$ is a total Hamiltonian defined in \cref{eq:total_hamiltonian}. For now, we do not make any assumptions about the mixer $H_{\text{init}}$. Recall that $\mathcal{F}$ in \cref{eq:feasible_set_binary} is a set of feasible candidate solutions to the 0--1 IP in \cref{eq:canonical_ip}. Let $\mathcal{H}_{\mathcal{F}}$ denote the feasible Hilbert subspace, and $\mathcal{H}_{\mathcal{Q}}$ its orthogonal complement. Then, the total Hilbert space $\mathcal{H}$ is partitioned as $\mathcal{H} = \mathcal{H}_{\mathcal{F}} \oplus \mathcal{H}_{\mathcal{Q}}$. Hence, we rewrite the eigenvalue equation in a block form and, for notational clarity, suppress the time variable $t$:
\begin{align*}
    H_{FF}\ket{\psi_F} + H_{FQ}\ket{\psi_Q} = E \ket{\psi_F}\\
    H_{QF}\ket{\psi_F} + H_{QQ}\ket{\psi_Q} = E \ket{\psi_Q}
\end{align*}
Here, $\ket{\psi_F} \coloneq \Pi_{\mathcal{F}} \ket{\psi(t)}$ and $\ket{\psi_Q} \coloneq \Pi_{\mathcal{Q}} \ket{\psi(t)}$ with the projector  $\Pi_{\mathcal{F}}$ defined in \cref{eq:projector_F} and $\Pi_{\mathcal{Q}} = I - \Pi_{\mathcal{F}}$. In other words, $\ket{\psi_F}$ and $\ket{\psi_Q}$ are blocks of $\ket{\psi(t)}$ that are supported on $\mathcal{H}_{\mathcal{F}}$ and $\mathcal{H}_{\mathcal{Q}}$, respectively. Similarly, $H_{AB} \coloneq \Pi_{\mathcal{A}} H(t) \Pi_{\mathcal{B}}$.

Then, using Schur complement and Taylor expansions in $\eta = \epsilon^2$ or $\eta = 1/\lambda$, we obtain the effective Hamiltonians governing transitions in the feasible subspace for penalty-based and penalty-free methods. The  effective Hamiltonian is:
\begin{equation}
    H_{\text{eff}}^{\eta} = H_{FF} - \frac{(1-s)^2}{ s}H^{\text{init}}_{FQ} \ \eta \, W  \ H^{\text{init}}_{QF} + O\left( \eta^2 \right).
\end{equation}
In the above, $s \equiv s(t)$, $H^{\text{init}}_{AB} := \Pi_{\mathcal{A}} H_{\text{init}} \Pi_{\mathcal{B}}$. The key operator that distinguishes the penalty-based and penalty-free methods is $\eta W$.

For penalty-based methods, we work in the large-penalty regime $\lambda \gg 1$. In this regime, the dominant energy scale of an infeasible mediator $\ket{z}$ is its penalty energy $\lambda g(z,w)$. Therefore, to leading order in $1/\lambda$, the transition mediated by $\ket{z}$ is weighted by
\begin{equation}
W_z^{\lambda} = \frac{1}{\lambda g(z,w)}.
\end{equation}
\begin{equation}
    \eta = \frac{1}{\lambda} \text{ and } W = \hat g^{-1},
\end{equation}
where $\hat g$ is a Hamiltonian encoding the penalty function $g(x,w)$ defined in \cref{eq:penalty_function}. We note that $\hat g$ is diagonal in the computational basis.
For the penalty-free methods, we let the mixer be $H_{\text{init}}^{\epsilon}$ defined in \cref{eq:relaxed_mixer}. Then, it follows that
\begin{equation}
    \eta = \epsilon^2 \text{ and } W = \left (  \hat f - \frac{E}{s} \right )^{-1},
\end{equation}
where, $\hat f$ is the Hamiltonian that encodes the objective function value $f(x) = c^{\mathsf{T}}x$ and $E \equiv E(t)$ is the instantaneous energy of the system. We note that $\hat f$ is diagonal in the computational basis.

We now examine how the effective Hamiltonian generates second-order state transitions. Since $H_{FF}$ is responsible for the first-order state transitions, we look at the second term of the effective Hamiltonian 
$$H^{\text{init}}_{FQ} \, \eta W \, H^{\text{init}}_{QF}.$$
For any feasible computational basis state $\ket{x} \in \mathcal{H}_{\mathcal{F}}$, the operator $H^{\text{init}}_{QF}$ creates a superposition of infeasible computational basis states $\ket{z}$ that are neighbours of $\ket{x}$. That is
\begin{equation*}
    H^{\text{init}}_{QF} \ket{x} = \sum_{\substack{z\in Q \\ z \sim x}} \alpha_z \ket{z}.
\end{equation*}
Then, the diagonal operator $\eta W$ weights each infeasible state $\ket{z}$. This yields
\begin{equation*}
    \eta W \, H^{\text{init}}_{QF} \ket{x} = \sum_{\substack{z\in Q \\ z \sim x}} \alpha_z W_z^{\eta } \ket{z},
\end{equation*}
where $W_z^{\eta}$ is given in \cref{eq:lambda_weight_factor} and \cref{eq:epsilon_weight_factor} for $\eta = 1/\lambda$ and $\eta = \epsilon^2$ respectively. Finally, we apply the operator $ H^{\text{init}}_{FQ}$, which completes a weighted second-order state transition to feasible states $\ket{y}$ that are neighbours of each $\ket{z}$. Hence, we have
\begin{equation*}
    H^{\text{init}}_{FQ} \, \eta W \, H^{\text{init}}_{QF} \ket{x} = \sum_{\substack{z\in Q \\ x \sim z}} \alpha_z W_z^{\eta } \sum_{\substack{y \in F \\ z \sim y}} \beta_y \ket{y}.
\end{equation*}
This proves how second-order state transitions are generated and weighted.
\end{proof}

\section{\label{sec:discussion} Outlook}
This work develops a general theory of quantum optimization for binary constrained problems. The central insight is that computational difficulty is governed by the restructuring of entanglement during evolution. Spectral gaps, Hamiltonian locality, and the amount of entanglement in the state are important diagnostics, but the dominant source of complexity is the creation, redistribution, and removal of entanglement required by the problem. In constrained combinatorial optimization, this restructuring is induced by the algebraic and combinatorial structure of the constraints.

This perspective refines the role of spectral gaps in quantum optimization. The minimum gap is often treated as the primary indicator of computational complexity. Here, the physical bottleneck is the entanglement restructuring required by the constrained problem. When the constraints force entanglement to be created, redistributed, or removed over a narrow interval of the evolution, the spectrum reflects this change through a narrowly avoided level crossing. Gap narrowing is therefore a spectral consequence of required entanglement restructuring. If the system remains on the same instantaneous energy level through the avoided crossing, it follows the corresponding eigenvector exchange and acquires the entanglement structure of the approaching level.

Narrowly avoided level crossings can also accelerate computation. As shown in \Cref{theorem:leakage_fast_jumps}, the constraints can create a sequence of narrowly avoided level crossings at which jumps are beneficial. At each crossing, the dynamical jump and the eigenvector swap occur together. The state moves to the approaching energy level while preserving its original eigenvector, thereby avoiding unnecessary entanglement restructuring. Thus, small gaps are not intrinsic obstructions. In the fast-jump regime analyzed here, they give a provable reduction in the required evolution time relative to adiabatic following through the weak avoided crossings.

Constraints should therefore be incorporated into the algorithmic dynamics rather than absorbed into generic penalty terms. Constraints determine the geometry of the feasible Hilbert space, the invariant subspaces of the Hamiltonian, the transition graph induced by the mixer, and the entanglement restructuring required during computation. Constraint-aware algorithms make these structures visible and controllable. Penalty-based reformulations can hide the algebraic and geometric information contained in the constraints and force unnecessary restructuring of the entanglement, thereby eliminating potential speedups.

This principle is closely aligned with classical mathematical optimization. Classical methods exploit constraints to reduce the effective search space. Linear programming, branch-and-bound, cutting-plane methods, and modern mixed-integer programming solvers use the geometry and combinatorial structure of constrained problems to isolate relevant candidates. In LP-based methods for solving 0--1 IPs, this refinement can be expressed schematically as
\begin{equation}\label{eq:nested_sets_outlook}
\mathcal{P} \supseteq \mathcal{P}_k \supseteq \mathcal{P}_{k+1} \supseteq \mathcal{F}.
\end{equation}
The relaxed feasible polyhedron is iteratively tightened while preserving the integer feasible set. The method exploits structure rather than blindly searching the original space.

A quantum analog is to refine the quantum routine itself. In the language developed here, this means refining the transition graph, the invariant subspaces, and the allowed restructuring of entanglement. One may begin with a restrictive mixer that enforces only part of the constraint system $Ax \geq b$ and induces minimal restructuring. Such a mixer may be fast, but it may also fragment the feasible subspace and make the optimum dynamically inaccessible. The algorithm can then be augmented by adding controlled transitions between previously invariant subspaces, or by allowing weak leakage through carefully chosen infeasible mediators. In this way, the quantum routine gains the flexibility needed to reach the optimum while avoiding unnecessary entanglement restructuring. Analogously to \cref{eq:nested_sets_outlook}, this gives a progressive sequence of refinements
\begin{equation}
    H_{\mathrm{init}}^{\mathrm{feas},0} \rightarrow H_{\mathrm{init}}^{\mathrm{feas},1} \rightarrow \dots \rightarrow H_{\mathrm{init}}^{\mathrm{feas},\ell},
\end{equation}
with corresponding transition graphs
\begin{equation}
    G^0 \subset G^1  \subset \dots  \subset G^\ell.
\end{equation}
The refinements above require separate runs with progressively more expressive mixers or feasibility-inducing measurements. By contrast, emerging feedback-based algorithms use records from intermediate or continuous measurements to construct circuits at runtime \cite{lemelin2025mid,magann2022feedback,gabbassov2026stochastic}.

This work does not focus on implementation details or hardware platforms. Its central question is structural: what entanglement restructuring is required by the constrained problem, and how can the algorithm reduce it? Any implementation, whether analog, digital, variational, or fault-tolerant, must ultimately address this question.

The long-term implication is that quantum optimization should become structurally aware. Classical optimization turns constraints from modelling obstacles into sources of efficient algorithmic search. Quantum optimization should do the same. The results presented here suggest that useful runtime advantages require algorithms that preserve, exploit, and progressively refine the algebraic and geometric properties of constrained problems. Entanglement restructuring provides the physical lens through which this refinement can be understood, quantified, and controlled.

\bf Acknowledgements \rm
EG acknowledges support through a grant from the National Research Council of Canada (NRC) and a Canada Graduate Scholarship from the National Science and Engineering Council of Canada (NSERC).

\appendix

\section{History of Linear Programming}\label{ap:history}
The developed understanding will serve as an important stepping stone toward its quantum analog. The field of modern optimization was born in the 1930s out of the need to manage economic resources and logistics. The Soviet mathematician and economist Leonid Kantorovich developed optimization problems with linear constraints for resource allocation under material shortages \cite{kantorovich1960mathematical}. His work remained largely unknown due to soviet secrecy and ideological reasons. However, in 1975, Kantorovich was awarded the Nobel Prize in economics. The field of optimization was further catalyzed by the Second World War, which required unprecedented logistics and resource management. It was George Dantzig who formulated Linear Programming in its modern standard form and invented the simplex algorithm, also known as ``the algorithm of the century" due to its extraordinary practical impact, theoretical depth and long-term economic and technological consequences \cite{bixby2012brief}. The field of linear programming experienced another rapid development during the era of the first electronic and then digital computers. Currently, all state-of-the-art linear programming solvers descend from the simplex method, directly or indirectly.

\section{Optimization Problems Used in the Main Text}
% \section{Benchmark Combinatorial Problems}
\label{app:benchmark_problems}

In this appendix, we specify the combinatorial optimization problems used to generate \Cref{fig:feas_mixer_max_indep_set,fig:spectral_gaps_knapsack}, and \Cref{fig:spectral_gaps_exact_cover}. These consist of an instance of the Maximum Independent Set Problem \cite[Ch.~4]{AvisHertzMarcotte2005} for \Cref{fig:feas_mixer_max_indep_set}, the 0--1 Knapsack Problem \cite[Ch.~17]{KorteVygen2012} for \Cref{fig:spectral_gaps_knapsack}, and the Exact Cover by 3-sets Problem (X3C) \cite[Ch.~7]{Knuth2020} for \Cref{fig:spectral_gaps_exact_cover}.

%For this, we considered three representative constrained binary problems. Fig.~\ref{fig:feas_mixer_max_indep_set} serves as an example of a transition graph induced by a feasibility-preserving mixer for the Maximum Independent Set Problem. The 0--1 Knapsack Problem, shown in
%Fig.~\ref{fig:spectral_gaps_knapsack}, is an instance of the penalty-free formulation in which feasibility is enforced through the mixer. Fig.\ref{fig:spectral_gaps_exact_cover} illustrates the Exact Cover by 3-sets Problem (X3C), which we use as an example of a penalty-based formulation together with a Hamming-weight preserving $XY$ mixer. 

\subsection{Maximum Independent Set Problem}\label{ap:maximum_independent_set}
Let $G=(V,E)$ be an undirected graph with vertices \mbox{$V=\{1,\dots,n\}$} and edge set $E$. An independent set is a subset of vertices that contains no adjacent pair. In other words, if two vertices $i$ and $j$ are connected by an edge \mbox{$(i,j)\in E$}, then they cannot both be selected.

We encode the choice of vertices by binary variables \mbox{$x_i\in\{0,1\}$}, where \mbox{$x_i=1$} means that vertex $i$ is selected and $x_i=0$ means that it is not selected. The condition that no two adjacent vertices are selected is then expressed by the constraints
\begin{equation}
    x_i+x_j \leq 1, \qquad (i,j) \in E.
\end{equation}
Thus, the feasible set of the maximum independent set problem is
\begin{equation}
    \mathcal{F} =\left\{ x \in \{0,1\}^n : x_i+x_j \leq 1 \text{ for all } (i,j)\in E \right\}.
\end{equation}
Each feasible bit string $x \in \mathcal{F}$ represents an independent set, namely the set of selected vertices
\begin{equation}
S(x)=\{ i \in V : x_i=1 \}.
\end{equation}

The maximum independent set problem asks for a feasible bit string that selects the largest number of vertices. 
\begin{equation}\label{eq:mis_ip}
    \begin{aligned}
    \text{max } \quad & \sum_{i=1}^n x_i \\
    \text{subject to } \quad & x_i+x_j \leq 1 , \text{ for all } (i,j)\in E, \\
    & x\in \{0,1\}^n .
\end{aligned}
\end{equation}
In the quantum setting, each bit string $x$ corresponds to a computational basis state $\ket{x}$. Feasible basis states $\ket{x}$ with $x\in\mathcal{F}$ represent independent sets of $G$. A feasibility-preserving mixer for the maximum independent set problem therefore induces transitions only between independent sets, so the dynamics never selects two adjacent vertices simultaneously. In \Cref{fig:feas_mixer_max_indep_set}, this construction is illustrated for the cycle graph $C_5$. The maximum independent sets of $C_5$ have size two, for example $10100$, $10010$, $01010$, $01001$, and $00101$, and these correspond to the optimal feasible computational basis states.

\subsection{0--1 Knapsack Problem}\label{ap:kp}
\Cref{fig:spectral_gaps_knapsack} uses an instance of the 0--1 Knapsack Problem to illustrate the penalty-free formulation. In this formulation, the constraints are implemented through the mixer by restricting evolution within a feasible subspace of the full binary search space. The Knapsack Problem provides a simple setting for this construction: its feasible set is specified by a single linear inequality, and its objective function is also linear.  It therefore serves as a simple constrained-binary problem for showing subspace invariance enforced through the mixer, and the resulting level crossings are discussed in the main text.

The 0--1 Knapsack Problem is a 0--1 integer program in which a subset of items is selected subject to a capacity constraint. An instance is specified by $n$ items with positive values $v_i>0$ and positive weights $w_i>0$, for $i=1,\ldots,n$, together with a positive capacity $C>0$. The goal is to choose a subset of items to put in the knapsack that maximizes the total value while keeping the total weight below the capacity $C$. Using binary decision variables $x_i\in\{0,1\}$, where $x_i=1$ means that item $i$ is selected, the problem can be written as:
\begin{equation}\label{eq:knapsack_ip}
    \begin{aligned}
        \text{max } \quad & \sum_{i=1}^n v_i x_i \\
        \text{subject to } \quad & \sum_{i=1}^n w_i x_i \le C\\
        & \ \ x \in \{0,1\}^n
    \end{aligned}
\end{equation}
Equivalently, in minimization form, we write
\begin{equation}\label{eq:knapsack_min_ip}
    \begin{aligned}
        \text{min } \quad & -\sum_{i=1}^n v_i x_i \\
        \text{subject to } \quad & \sum_{i=1}^n w_i x_i \le C\\
        & \ \ x \in \{0,1\}^n
    \end{aligned}
\end{equation}

The precise knapsack instance used for \Cref{fig:spectral_gaps_knapsack} has five items ($n=5$) with values \mbox{$v=(8,1,9,3,5)$}, weights \mbox{$w=(2,2,5,7,1)$} and capacity $C=8$.
The optimal feasible solution for this problem is \mbox{$x^\star=(1,0,1,0,1)$}, which has total weight $w^{\mathsf T} x^\star=8$ and objective function value $f(x^\star)=-22.$

Importantly, there are infeasible configurations with lower objective function values than the optimal solution. In particular, the three infeasible solutions $(1,1,1,0,1),(1,0,1,1,1)$ and $(1,1,1,1,1)$ have objective function values $-23,-25$ and $-26$ respectively, all of which lie below the optimal objective function value $f(x^\star)=-22$. Consequently, the eigenvalue branch that starts in the ground state of the mixer and terminates at the optimal solution must cross these three low-energy infeasible levels. For $\epsilon=0$, corresponding to the feasible-subspace mixer, the level crossings are exact [see \Cref{fig:spectral_gaps_knapsack} (a)]. For $\epsilon>0$, the relaxed mixer weakly couples the otherwise invariant subspaces, and the same crossings become narrowly avoided level crossings [see \Cref{fig:spectral_gaps_knapsack} (b)]. These infeasible low-energy states are the red levels shown below the optimal feasible level in \Cref{fig:spectral_gaps_knapsack}. In the figure, the spectrum of the problem Hamiltonian has been scaled and shifted vertically for clarity.

\subsection{Exact Cover by 3-Sets}\label{ap:x3c}
\Cref{fig:spectral_gaps_exact_cover} uses an instance of the Exact Cover by 3-Sets Problem (X3C) to illustrate that the penalty-based formulation is a spectral dual of the penalty-free case. The structure of X3C provides prior knowledge about the Hamming weight of the optimal solution. This prior information suggests choosing a mixer that preserves Hamming weight, which restricts the evolution to an invariant fixed-Hamming-weight subspace. This subspace invariance results in level crossings shown in \Cref{fig:spectral_gaps_exact_cover} (a).

The exact cover by 3-sets problem, or X3C, is a constrained binary problem. One is given a set $U$ with cardinality $|U| = N$ and a collection $\mathcal S$ of 3-element subsets $S_j$ such that
\begin{equation}
    \mathcal S=\{S_1,\ldots,S_n\},
    \qquad
    S_j\subseteq U,
    \qquad
    |S_j|=3.
\end{equation}
The task is to choose an exact cover of $U$, i.e, to choose $q$ of these $n$ subsets, whose union is $U$, with every element of $U$ covered exactly once. Equivalently, one seeks indices $i_1,\ldots,i_r$ such that
\begin{equation}
\bigcup_{\ell=1}^{r} S_{i_\ell} = U,
\ \
S_{i_\ell}\cap S_{i_m}=\varnothing
\ \
\text{ for all } \ell\neq m.
\end{equation}
Since each subset contains three elements and the selected subsets are pairwise disjoint, an exact cover must contain $q=N/3$ subsets.

Introducing binary variables $x_j\in\{0,1\}$, where $x_j=1$ means that subset $S_j$ is selected, the exact-cover problem can be written as:
% \begin{equation}
    % \sum_{j:\,a\in S_j}x_j=1,
    % \qquad
    % a\in U.
% \end{equation}
\begin{equation}
\label{eq:x3c_ip}
    \begin{aligned}
        \text{min}\quad & f(x)=0 \\
        \text{subject to}\quad
        & \sum_{j:\,a\in S_j} x_j = 1,
        \qquad \forall\, a\in U,\\
        & x\in\{0,1\}^n .
    \end{aligned}
\end{equation}
These constraints enforce that each element $a$ is covered by exactly one selected subset. In addition, the structure of the problem admits an implicit constraint that the optimal solution must have a Hamming weight of $|U|/3$ or equivalently, \mbox{$\sum_{i=1}^{n} x_i = |U|/3$}. The zero objective function value means that all feasible exact covers are optimal. In the instance considered here, the exact cover is unique.

The precise X3C instance used for \Cref{fig:spectral_gaps_exact_cover} has a universe \mbox{$U=\{1,\ldots,6\}$} and eight 3-element subsets,
\begin{equation}
\begin{aligned}
S_1&=\{1,2,3\}, S_2=\{4,5,6\},
S_3=\{1,2,4\}, S_4=\{1,2,5\}, \\
S_5&=\{1,3,6\},  S_6=\{1,4,6\}, 
S_7=\{2,4,6\}, S_8=\{3,4,5\}.\nonumber
\end{aligned}
\end{equation}
Since \mbox{$\lvert U\rvert=6$}, an exact cover must contain $q=6/3=2$ subsets and thus, the optimal solution has Hamming weight 2. The unique exact cover is \mbox{$\{S_1,S_2\}$}, and so the optimal solution for this problem is \mbox{$x^\star~=~(1,1,0,0,0,0,0,0)$}.

For the complete-graph $XY$ mixer used in the figure, the third eigenstate of the full mixer is degenerate. One of the degenerate eigenstates is the ground state of the mixer restricted to the Hamming-weight-2 sector. We choose this state as the initial state. It is an excited eigenstate of the full mixer, while $\ket{x^\star}$ is the ground-state of $H_{\mathrm p}$. The associated eigenvalue trajectory therefore descends from the third mixer level to the ground level of the problem Hamiltonian. This eigenvalue trajectory crosses three penalized suboptimal levels exactly, each of which terminates above the optimal solution at $s=1$ [see \Cref{fig:spectral_gaps_exact_cover} (a)]. The spectrum of the problem Hamiltonian has been scaled and shifted vertically for clarity.

To illustrate the possibility of avoided crossings, we weakly relax the strict Hamming-weight conservation of the $XY$ mixer. For \Cref{fig:spectral_gaps_exact_cover} (b), we add a transverse-field Hamiltonian perturbation projected onto the subspace spanned by the first five eigenstates of the $XY$ mixer shown in the figure. So, the total mixer looks like
\begin{equation}
    H_{\mathrm{init}} = H^{\mathrm{XY}}_{\mathrm{init}} +\epsilon \Pi_{\mathrm{low}} \left(-\sum_{i=1}^n X_i\right) \Pi_{\mathrm{low}},
\end{equation}
where \mbox{$\Pi_{\mathrm{low}} =\sum_{k=1}^{5} \Pi_k$} and $\Pi_k$ being the projector onto $k$th eigenstate of $H^{\mathrm{XY}}_{\mathrm{init}}$. This is a specific construction of a mixer that weakly couples otherwise distinct Hamming-weight sectors. The exact crossings of \Cref{fig:spectral_gaps_exact_cover} (a) are thereby converted into narrowly avoided crossings.

\section{Hamiltonian Encoding}\label{ap:hamiltonian_encoding}
This section describes how constrained binary optimization problems are encoded into quantum Hamiltonians.

To encode the problem for quantum optimization, each binary string $x\in\{0,1\}^n$ is associated with the computational basis state $\ket{x}$ in the $n$-qubit Hilbert space of dimension $2^n$. The objective function $f(x)$ [\Cref{eq:canonical_ip}] or $q(x)$ [\Cref{eq:unconstrained}] is then encoded as the spectrum of a problem Hamiltonian that is diagonal in the computational basis. Thus, we have
\begin{equation}
    H_{\mathrm{p}}\ket{x} = g(x)\ket{x}, \quad g(x) \in \{f(x), q(x)\}.
\end{equation}
Minimizing $g(x)$ corresponds to finding the lowest-energy state within the feasible subspace. 

To map the objective function $g(x)$ to a Hamiltonian, we perform the following transformation:
\begin{equation}
    x_i \rightarrow \frac{I-Z_i}{2}
\end{equation}
Here, $x_i$ is the $i$th variable of the vector $x$, and $Z_i$ is the Pauli-$Z$ operator acting on qubit $i$.

For example, mapping a linear function $f(x)$ to a Hamiltonian yields:
\begin{equation}
    f(x)=-\sum_{i=1}^n v_i x_i
    \ \longrightarrow \
    H_{p}
    = -\sum_{i=1}^n v_i\frac{I-Z_i}{2}
\end{equation}
The problem Hamiltonian can be rewritten as:
\begin{equation}
    H_{\mathrm{p}} =- \frac{1}{2}\sum_{i=1}^n v_i I + \frac12\sum_{i=1}^n v_i Z_i
\end{equation}
The term proportional to identity only shifts the spectra by a constant and thus is usually omitted. Therefore, $H_{\mathrm p}$ is diagonal in the computational basis, with computational basis states $\ket{x}$ as eigenvectors and objective values $f(x)$ as the corresponding eigenvalues.

We now discuss \Cref{fig:spectral_gaps_knapsack}, which illustrates the penalty-free approach applied to a knapsack instance. In this setting, the problem Hamiltonian $H_{\mathrm p}$ encodes only the linear objective function $f(x)$; the capacity constraint is not incorporated as an
energetic penalty. Instead, feasibility is enforced through the mixer. %Let $\Pi_{\mathcal F}$ denote the projector onto the feasible subspace and let $\Pi_{\mathcal Q}=I-\Pi_{\mathcal F}$ denote the projector onto its orthogonal complement. 
Starting from the transverse-field mixer, \mbox{$H_{\mathrm{init}}=-\sum_{i=1}^n X_i$}, we use the relaxed feasible-subspace mixer introduced in
\cref{eq:relaxed_mixer}, with the parameter $\epsilon$ controlling the strength of leakage between the feasible and infeasible subspaces. Explicitly, the mixer used for the knapsack instance is
\begin{equation}
    H^{\epsilon}_{\mathrm{init}}
    =
    -(\Pi_{\mathcal F}+\epsilon\Pi_{\mathcal Q})
    \left(\sum_{i=1}^n X_i\right)
    (\Pi_{\mathcal F}+\epsilon\Pi_{\mathcal Q}) .
\end{equation}
For $\epsilon=0$, the feasible and infeasible subspaces are exactly invariant; the mixer only connects feasible bit strings that differ by a single bit flip, and all infeasible computational basis states lie in the kernel. So, for $\epsilon=0$:
\begin{equation}
    H_{\mathrm{init}}^{\epsilon}\ket{z}=0,
    \qquad
    z \notin \mathcal F
\end{equation}
The total Hamiltonian used for the interpolation is
\begin{equation}
    H(t)
    =(1-s(t))H_{\mathrm{init}}^{\epsilon}+
    s(t)H_{\mathrm{p}}.
\end{equation}

\bibliography{refs}% Produces the bibliography via BibTeX.

\end{document}